\documentclass{article}
\pdfoutput=1
\usepackage{arxiv}
\usepackage{ifpdf}
\usepackage[utf8]{inputenc} 
\usepackage[T1]{fontenc}    
\usepackage{hyperref}       
\usepackage{url}            
\usepackage{booktabs}       
\usepackage{amsfonts}       
\usepackage{nicefrac}       
\usepackage{microtype}      
\usepackage{lipsum}
\usepackage{amsmath,bm}
\usepackage{graphicx}
\usepackage{array}
\usepackage{caption}
\usepackage{booktabs, makecell, multirow, threeparttable}
\usepackage{float}

\newtheorem{theorem}{Theorem}[section]
\newtheorem{corollary}{Corollary}[theorem]
\newtheorem{assumption}{Assumption}[theorem]

\newtheorem{rem}[theorem]{Remark}
\newtheorem{definition}{Definition}[section]
\newtheorem{proof}{Proof}[section]
\newcommand{\RomanNumeralCaps}[1]{\MakeUppercase{\romannumeral #1}}
\newcommand{\norm}[1]{\left\lVert#1\right\rVert}

\title{Multivariate General Compound Point Processes in Limit Order Books}

\author{
  Qi Guo\\
  Department of Mathematics and Statistics\\
  University of Calgary\\
  University Drive NW\\
  Calgary, Canada T2N 1N4\\
  \texttt{qi.guo1@ucalgary.ca} \\
   \And
 Bruno Remillard \\
  Department of Decision Sciences\\
  HEC Montr\'{e}al\\
  3000, chemin de la Cote-Sainte-Catherine\\
  Montr\'{e}al, Canada H3T 2A7\\
  \texttt{bruno.remillard@hec.ca} \\
  \AND
   Anatoliy Swishchuk\\
  Department of Mathematics and Statistics\\
  University of Calgary\\
  University Drive NW\\
  Calgary, Canada T2N 1N4\\
  \texttt{aswish@ucalgary.ca} \\
}

\begin{document}
\maketitle

\begin{abstract}
In this paper, we focus on a new generalization of multivariate general compound Hawkes process (MGCHP), which we referred to as the multivariate general compound point process (MGCPP). Namely, we applied a multivariate point process to model the order flow instead of the Hawkes process. Law of large numbers (LLN) and two functional central limit theorems (FCLTs) for the MGCPP were proved in this work. Applications of the MGCPP in the limit order market were also considered. We provided numerical simulations and comparisons for the MGCPP and MGCHP by applying Google, Apple, Microsoft, Amazon, and Intel trading data.
\end{abstract}


\section{Introduction}
In this paper we study multivariate general compound point processes to model the price processes in the limit order books (LOB). We prove a Law of Large Numbers and Functional Central Limit Theorems (FCLT) for these processes. The latter two FCLTs are applied to limit order books where we use these asymptotic methods to study the link between price volatility and order flow in our two models by using the diffusion limits of these price processes. The volatilities of price changes are expressed in terms of parameters describing the arrival rates and price changes.  

Bacry {\it et al.} (2013) proved a LLN and FCLT for multivariate HP \cite{Bacry2013}. Bowsher (2007) was the first who applied a HP and point processes to financial data modelling \cite{Bowsher2007}. Bauwens and Hautsch (2009) use a 5-D HP to estimate multivariate volatility, between five stocks, based on price intensities \cite{Bauwens2009}.  We note, that Br\'{e}maud {\it et al.} (1996) generalized the HP to its nonlinear form \cite{Bremaud1996}. Also, a functional central limit theorem for the nonlinear HP was obtained in \cite{Zhu2013}. Some applications of multivariate HP to financial data are  given in \cite{Embrechts2011}.
Vinkovskaya (2014)  considered a point process model for the dynamics of LOB, and a regime-switching HP to model its dependency on the bid-ask spread in limit order books \cite{Vinkovskaya2014}.  A semi-Markov process was applied to LOB in \cite{Swishchuk20172} to model the mid-price. We note, that a level-1 limit order books with time dependent arrival rates $\lambda(t)$ were studied in \cite{Chavez2017}, including the asymptotic distribution of the price process. General semi-Markovian models for limit order books were considered in \cite{Swishchuk20173}. The book by Cartea {\it et al.} (2015) develops models for algorithmic trading in contexts such as executing large orders, market making, trading pairs or collecting of assets, and executing in dark pool \cite{Cartea2015}. That book also contains link to a website from which many datasets from several sources can be downloaded, and MATLAB code to assist in experimentation with the data. A detailed description of the mathematical theory of Hawkes processes is given in \cite{Liniger2009}.  Zheng {\it et al.} (2014) introduced a multivariate point process describing the dynamics of the Bid and Ask price of a financial asset \cite{Zheng2014}. The point process is similar to a Hawkes process, with additional constraints on its intensity corresponding to the natural ordering of the best Bid and Ask prices.  Eichler {\it et al.}  (2017) has shown that the Granger causality structure of multivariate HP is fully encoded in the corresponding link functions of the model \cite{Eichler2017}. A new nonparametric estimator of the link functions based on a time-discretized version of the point process was introduced by using an infinite order autoregression. Consistency of the new estimator was derived. The estimator was applied to simulated data and to neural spike train data from the spinal dorsal horn of a rat. 
Chen {\it et al.}  (2019)  developed a new approach for investigating the properties of the HP without the restriction to mutual excitation or linear link functions \cite{Chen2019}. They employed a thinning process representation and a coupling construction to bound the dependence coefficient of the HP. Using recent developments on weakly dependent sequences, a concentration inequality for second-order statistics of the HP was established. This concentration inequality was applied to cross-covariance analysis in the high-dimensional regime, and it was verified the theoretical claims with simulation studies \cite{Chen2019}.  Lemonnier {\it et al.}  presented a framework for fitting multivariate HP for large-scale problems, both in the number of events in the observed history $n$ and the number of event types $d$ (i.e. dimensions) \cite{Lemonnier2017}. 
Liniger (2009) thesis addresses theoretical and practical questions arising in connection with multivariate, marked, linear HP \cite{Liniger2009}. 
Yang {\it et al.}  (2017)  developed a nonparametric and online learning algorithm that estimates the triggering functions of a multivariate HP \cite{Yang2017}. \cite{Rambaldi2017} has shown that multivariate Hawkes processes coupled with the nonparametric estimation procedure can be successfully used to study complex interactions between the time of arrival of orders and their size observed in a limit order book market. This methodology was applied to high-frequency order book data of futures traded at EUREX. 
Introduction to point processes from a martingale point of view may be found in Bjork (2011) lecture notes \cite{Bjork2011}.

Guo {\it et al.} (2020) constructed a multivariate general compound Hawkes process (MGCHP) \cite{Guo2020} which is an extended model from \cite{Cont2013} and \cite{Swishchuk20171}. In \cite{Guo2020}, they applied the multivariate Hawkes process to model the order flow of several stocks in limit order market and proved limit theorems for the MGCHP. In this paper, we proposed a new mid-price model which is a generalization of the MGCHP and we called it the multivariate general compound point process (MGCPP). For the MGCPP, we applied a multi-dimensional simple point process to represent the order flow in LOB instead of the Hawkes process. We also proved the corresponding LLN and FCLTs for the MGCPP. One of the reasons why we considered the generalized model is parameters for simple point process are much easier to estimate than Hawkes process. So, we provided the numerical comparisons of the MGCPP and MGCHP by real high-frequency trading data and we found that results of the new generalized model are as good as the MGCHP. 

This paper is organized as follows. Definition and assumptions of the multivariate general compound point process (MGCPP) can be found in Section \ref{sec2}. Functional central limit theorem (FCLT) \RomanNumeralCaps{1} and law of large numbers were proved in Section \ref{sec3}. We also provided numerical examples simulated by real data for the FCLT \RomanNumeralCaps{1} in Section \ref{sec3}. In Section \ref{sec4}, we considered a FCLT \RomanNumeralCaps{2} for the MGCPP and applied it in the mid-price prediction. Section \ref{sec5} concludes the paper.

\section{Definition of Multivariate General Compound Point Process (MGCPP)} \label{sec2}

In this Section, we proposed a multivariate stochastic model for the mid-price in the limit order book. This is a generalization for models in \cite{Cont2013}, \cite{Guo2020}, and \cite{Swishchuk20171}. Here, we assume the order flow was described by a multivariate simple point process with some good asymptotic properties. 

\begin{definition}
	(Counting Process). (see, eg., \cite{Daley2007}): We called a stochastic process $\{N(t), t \geq 0\}$ counting process if it satisfies $N(t) \geq 0 $, $N(0) = 0$, $N(t+s) \geq N(t)$, for all $t,s \geq 0$, and $N(t)$ is an integer. 
\end{definition} 

\begin{definition}
	(Point Process). (see, eg., \cite{Daley2007}): Let $(T_1, \, T_2, \, T_3, \cdots )$ be a sequence of non-negative random variables with $P(0 \leq T_1 \leq T_2 \leq T_3 \leq \cdots) = 1$, and the number of points in a bounded region is almost surely finite, then $(T_1, \, T_2, \, T_3, \cdots )$ is called a point process.	
\end{definition}

The point process was characterized by the conditional intensity function $\lambda (t)$ in the form of
\begin{equation}
\lambda (t) = \lim_{h \rightarrow 0} \frac{E[N(t+h)-N(t)|{\cal{F}}^N (t)]}{h},
\end{equation}
where $\lambda (t)$ is a non-negative function and ${\cal{F}}^N (t)$, $t > 0$ is the corresponding natural filtration. 

\subsection{Assumptions for Multivariate Point Processes}
Let $ \Vec{N}_t = (N_{1,t}, N_{2,t}, \cdots, N_{d,t},)$ be $d$-dimensional point process with following assumptions: 

\begin{assumption} \label{as1}
	We assume there's a law of large numbers (LLN) of the $\Vec{N}_t$ in the form of:
	
	\begin{equation} \label{as1f}
	\frac{\Vec{N}(nt)}{n} \rightarrow \Vec{\bar{\lambda}}t
	\end{equation}	
	as $n \rightarrow+\infty$ almost-surely, where $\Vec{\bar{\lambda}} = (\bar{\lambda}_1, \bar{\lambda}_2, \bar{\lambda}_3, \cdots, \bar{\lambda}_d) $.
\end{assumption}

\begin{assumption} \label{as2}
	We also assume there's a Functional Central Limit Theorem (FCLT) of the $\Vec{N}_t$ in the form of:
	
	\begin{equation} \label{as2f}
	\frac{1}{\sqrt{n}}(\Vec{N}_{nt} - E(\Vec{N}_{nt})), \,\, t \in [0,1]
	\end{equation} 
	converge in law of the Skorohod topology to $\Sigma^{1/2}\Vec{W}_t$ as $n \rightarrow \infty$, where $\Vec{W}_t$ is a standard $d$-dimensional Brownian motion and $\Sigma$ is in the form of:
	$\Sigma = \mathrm{diag} (\sigma^2_1,\, \sigma^2_2,\, \sigma^2_3, \,\cdots, \sigma^2_d)$.
\end{assumption}

Here, $\Vec{N_t}$ denotes the order flow in the limit order market for $d$ stocks. Liquidity for the high-frequency trading data guarantee there are enough price changes in one day or even a small window size $nt$. So, it is resealable to consider those two limit assumptions before.   

\begin{rem}
	
	For a simple example, if we consider the point process as a multivariate homogeneous Poisson process, then two assumptions above are LLN and FCLT for the multi-dimensional Poisson process. Let $\vec{P}_t$ be a $d$-dimensional Poisson process with intensity $\vec{\lambda}$. Here, we used notation $\vec{P}_t$ to distinguish the general case and Poisson example. Then, we have the LLN in the form of
	\begin{equation}
	sup_{t \in [0,1]} \norm{ n^{-1} \vec{P}_{nt} - t\vec{\lambda}} \rightarrow 0
	\end{equation}
	as $n \rightarrow \infty$ almost-surely. And the FCLT in the form of 
	\[
	\sqrt{n} \bigg(  \frac{1}{n} \vec{P}_{nt} -t\vec{\lambda} \bigg)
	\]
	converge in law for the Skorokhod topology to $  \vec{W}_t \circ \vec{\lambda}^{1/2}$ as $n \rightarrow \infty$, where $\circ$ is the element-wise product. 
\end{rem} 

\begin{rem}
	Another interesting example is the multivariate Hawkes process (MHP). Let $\vec{H}_t = (H_{1,t},\, H_{2,t}, \, \cdots, H_{d,t})$ be a $d$-dimensional Hawkes process with the intensity function for each $H_i$ in the form of
	\begin{equation}
	\lambda_i(t) = \lambda_i + \int_{(0,t)} \sum_{j=1}^d \mu_{ij}(t-s) d H_{j,s},
	\end{equation}
	Let $\boldsymbol{\mu} = (\mu_{ij})_{1 \leq i,j \leq d}$, $\vec{\lambda} = (\lambda_1, \lambda_2, \cdots, \lambda_d)^T$, and $\mathbf{K} = \int_0^{\infty} \boldsymbol{\mu} (t) dt$, then the LLN for MHP is in the form of
	\begin{equation}
	sup_{t \in [0,1]} \norm{ n^{-1} \vec{H}_{nt} - t(\mathbf{I}-\mathbf{K})^{-1}\vec{\lambda}} \rightarrow 0
	\end{equation}
	as $n \rightarrow \infty$ almost-surely, where $\mathbf{I}$ is a $d$-dimensional identity matrix. And we can also have the FCLT for MHP:	
	\[
	\frac{1}{\sqrt{n}}(\vec{H}_{nt} - E(\vec{H}_{nt})), \,\, t \in [0,1]
	\] 
	converge in law of the Skorohod topology to $(\mathbf{I} - \mathbf{K})^{-1} \mathbf{D}^{1/2}\vec{W}_t$ as $n \rightarrow \infty$, where $\vec{W}_t$ is a standard $d$-dimensional Brownian motion and $\mathbf{D}$ is a diagonal matrix such that $\mathbf{D}_{ii} = ((\mathbf{I}-\mathbf{K})^{-1}\vec{\lambda})_i$. Details about the LLN and FCLT of MHP can be found in \cite{Bacry2013}.
\end{rem}

\subsection{Definition for MGCPP}

Next, we consider a price process $ \vec{S_t}$ in the form
$\vec{S}_t = (S_{1,t}, S_{2,t}, \cdots, S_{d,t},)$ as:
\begin{equation} \label{priceprocess}
S_{i,t} = S_{i,0} + \sum_{k=1}^{N_{i,t}} a_i(X_{i,k}),
\end{equation}
where $X_{i,k}$ are independent ergodic continuous-time Markov chains and $a_i(\cdot)$ are bounded continuous functions on $X$. We refer $S_t$ as multivariate general compound point processes (MGCPP).

\begin{rem}
	If we consider the one-dimensional case, let $N_t$ be a Poisson process, $a(x) = x $, and $X_k$ is a sequence of
	independent random variables such that $P(X_1=\delta)=P(X_1=-\delta)=1/2$, then $S_t$ is a stochastic model for the dynamics of a limit order book discussed in \cite{Cont2013}. 
\end{rem}

\begin{rem}
	When $\vec{N}_t$ is a multivariate Hawkes process, then $\vec{S}_t$ is a multivariate general compound Hawkes processes (MGCHP) which proposed in \cite{Guo2020}.
\end{rem}

\section{LLNs and Diffusion Limits for MGCPP}\label{sec3}

In this Section, we considered the diffusion limit theorems for the MGCPP. It provides us a link between the order flow $\Vec{N}_t$ and the price process $\Vec{S}_t$. The functional central limit theorem and law of large numbers for the MGCPP are generalizations for the diffusion limit theorems of the MGCHP in \cite{Guo2020}.

\subsection{LLN for MGCPP}

\begin{theorem} \label{lln} (LLN for MGCPP). Let $\Vec{S}_{nt} = (S_{1,nt},S_{2,nt},S_{3,nt}, \cdots, S_{d,nt})$ be a $d$-dimensional general compound point process defined before, we have	
	\[
	\frac{\Vec{S}_{nt}}{n} \rightarrow \Tilde{a^*} \Vec{\bar{\lambda}}t
	\]
	as $n \rightarrow \infty$ almost-surly.
\end{theorem} 

\begin{proof} [Proof of Theorem \ref{lln}]
	From the definition of MGCPP in equation (\ref{priceprocess}), we have 	
	\[
	\frac{S_{i, n t}}{n}=\frac{S_{i,0}}{n}+\sum_{k=1}^{N_{i, n t}} \frac{a_i \left(X_{i, k}\right)}{n}.
	\]
	Since $S_{i,0}$ is a constant, we have
	\begin{equation} 
	\begin{split}
	\lim_{n\rightarrow \infty} \left( \frac{S_{i,t}}{n} \right) &= \lim_{n\rightarrow \infty} \left( \frac{S_{i,0}}{n} \right)
	+   \lim_{n\rightarrow \infty} \frac{ \sum_{k=1}^{N_{i,nt}} a_i(X_{i,k})}{n}  \\ & = 0 + \lim_{n\rightarrow \infty}  \frac{ \sum_{k=1}^{N_{i,nt}} a_i(X_{i,k})}{n}.
	\end{split}
	\end{equation}
	Recall the strong LLN of Markov chain (see, eg,. \cite{Norris1998}), we have 	
	\[
	\frac{1}{n} \sum_{k=1}^{n} a_i\left(X_{i, k}\right) \rightarrow_{n \rightarrow+\infty} a_{i}^{*}, \,\,\,\,\, a.s.,
	\] 
	where $a_i^*$ is defined by $a_i^* = \sum_{k \in X_i} \pi_{i,k}^{*} a_i\left(X_{i,k}\right)$. Consider the LLN of MPP in assumption \ref{as1}, we have
	\[
	\frac{N_{i,nt}}{n} \rightarrow {\bar{\lambda}_i}t
	\]
	as $n \rightarrow \infty$ almost-surly, we obtain	
	\begin{equation} \label{llnpf}
	\frac{1}{n} \sum_{k=1}^{N_{i, n t}} a_i\left(X_{i, k}\right)=\frac{N_{i, n t}}{n} \frac{1}{N_{i, n t}} \sum_{k=1}^{N_{i, n t}} a_i\left(X_{i, k}\right) \rightarrow_{n \rightarrow+\infty} a_{i}^{*} \bar{\lambda_i}t, \,\,\,\,\, a.s.
	\end{equation}	
	Rewrite (\ref{llnpf}) in the multivariate case, we derive the LLN for the MGCPP.	
\end{proof}

\subsection{Diffusion Limits for MGCPP: Stochastic Centralization}

\begin{theorem} \label{FCLT1} (FCLT \RomanNumeralCaps{1}: Stochastic Centralization). Let $X_{i,k}, \, i=1,2,\cdots,d$ be independent ergodic Markov chains with $n$ states $\{1,2,\cdots,n\}$ and with ergodic probabilities $\left(\pi_{i,1}^{*}, \pi_{i,2}^{*}, \ldots, \pi_{i,n}^{*}\right)$. Let $\vec{S}_{nt}$ be $d$-dimensional general compound point process, we have 
	\begin{equation} \label{TFCLT}
	\frac{\vec{S}_{nt} - \Tilde{a^*} \vec{N}_{nt}}{\sqrt{n}} \longrightarrow \Tilde{\sigma}^* \Lambda^{1/2} \Vec{W}(t), \, for\, all\, t>0 
	\end{equation}   
	as $n \rightarrow \infty$, where $\Vec{W}(t)$ is a standard $d$-dimensional Brownian motion, $\Lambda$ is a diagonal matrix such that $\Lambda = \normalfont{diag} (\bar{\lambda}_1,\, \bar{\lambda}_2,\, \bar{\lambda}_3, \,\cdots, \bar{\lambda}_d) $, $\vec{N}_{nt}$ is a $d$-dimensional vector, $\Tilde{a^*}$ and $\Tilde{\sigma}^*$ are diagonal matrices
	\[
	\Tilde{a^*} =
	\begin{bmatrix}
	a^*_1 & \cdots & 0\\
	\vdots& \ddots & \vdots\\
	0 & \cdots & a^*_d
	\end{bmatrix}, \,   \vec{N}_{nt} =
	\begin{bmatrix}
	N_{1,nt}  \\
	\vdots \\
	N_{d,nt}
	\end{bmatrix}, \,\Tilde{\sigma}^* =
	\begin{bmatrix}
	\sigma^*_1 &\cdots & 0 \\
	\vdots & \ddots & \vdots \\
	0 & \cdots & \sigma^*_d
	\end{bmatrix}.
	\]
	
	Here, $a_i^* = \sum_{k \in X_i} \pi_{i,k}^{*} a_i\left(X_{i,k}\right)$, and $\left(\sigma^{*}_i\right)^{2} :=\sum_{k \in X_i} \pi_{i,k}^{*} v_i(k)$ with
	\[
	\begin{aligned} v_i(k) &=b_i(k)^{2}+\sum_{j \in X_i}(g_i(j)-g_i(k))^{2} P_i(k, j) - 2 b_i(k) \sum_{j \in X_i}(g_i(j)-g_i(k)) P_i(k, j) \\ b_i &=(b_i(1), b_i(2), \ldots, b_i(n))^{\prime} \\ b_i(k) : &=a_i(k)-a_i^{*} \\ g_i : &=\left(P_i+\Pi_i^{*}-I\right)^{-1} b_i, \end{aligned}
	\]
	where $P_i$ is the  transition probability matrix for the Markov chain $X_i$, $\Pi_i^*$ is the matrix of stationary distributions of $P_i$,
	and $g_i(j)$ is the $j$th entry of $g_i$.	
\end{theorem}

\begin{proof} [Proof of Theorem \ref{FCLT1}]	
	From the definition of MGCPP, we have	
	\begin{equation} 
	S_{i,nt} = S_{i,0} + \sum_{k=1}^{N_{i,nt}} a_i(X_{i,k}),
	\end{equation}
	and
	\begin{equation} 
	S_{i,t} = S_{i,0} + \sum_{k=1}^{N_{i,nt}} (a_i(X_{i,k}) - a_i^*) + a_i^* N_{i,nt},
	\end{equation}
	here the $a_i^*$ is defined by $a_i^* = \sum_{k \in X_i} \pi_{i,k}^{*} a_i\left(X_{i,k}\right)$. Then, for some $n$, we have
	
	\begin{equation} 
	\frac{S_{i,t} - a_i^* N_{i,nt}}{\sqrt{n}}=  \frac{S_{i,0} + \sum_{k=1}^{N_{i,nt}} (a_i(X_{i,k}) - a_i^*)}{\sqrt{n}}.
	\end{equation}
	Since $S_{i,0}$ is a constant, when $n \rightarrow \infty$, we have
	\begin{equation} 
	\begin{split}
	\lim_{n\rightarrow \infty} \left( \frac{S_{i,t} - a_i^* N_{i,nt}}{\sqrt{n}} \right) &= \lim_{n\rightarrow \infty} \left( \frac{S_{i,0}}{\sqrt{n}} \right)
	+   \lim_{n\rightarrow \infty} \left( \frac{ \sum_{k=1}^{N_{i,nt}} (a_i(X_{i,k}) - a_i^*)}{\sqrt{n}}  \right)\\ & = 0 + \lim_{n\rightarrow \infty} \left( \frac{ \sum_{k=1}^{N_{i,nt}} (a_i(X_{i,k}) - a_i^*)}{\sqrt{n}}  \right).
	\end{split}
	\end{equation}
	
	Consider the following sums:  
	\[ R_{i, n}^{*}:=\sum_{k=1}^{n}\left(a_i\left(X_{i, k}\right)-a_{i}^{*}\right), \]
	and
	\[
	U_{i, n}^{*}(t):=n^{-1 / 2}\left[(1-(n t-\lfloor n t\rfloor)) R_{i,\lfloor n t\rfloor}^{*}+(n t-\lfloor n t\rfloor) R_{i,\lfloor n t\rfloor+1}^{*}\right],
	\]
	where $\lfloor \cdot \rfloor $ is the floor function. As the similar martingale method in \cite{Swishchuk20172} and \cite{Vadori2015}, we have the following weak convergence in Skorokhod topology 
	\begin{equation}\label{conv}
	U_{i, n}^{*}(t) \rightarrow_{n \rightarrow+\infty} \sigma_{i}^{*} W_{i}(t).
	\end{equation}
	
	From the assumption (\ref{as1}), we have the LLN for the MPP in the form of
	
	\[
	\frac{N_{i}(nt)}{n} \rightarrow_{n \rightarrow \infty} \bar{\lambda}_i t.
	\]  
	
	Using change of time in (\ref{conv}) and let $t \rightarrow N_i(nt)/n $, we have
	
	\begin{equation} \label{conv1}
	U_{i, n}^{*}(N_i(nt)/n) \rightarrow_{n \rightarrow+\infty} \sigma_{i}^{*} \sqrt{\bar{\lambda_i}}W_{i}(t).
	\end{equation}
	
	Rewrite (\ref{conv1}) in the multivariate form we derive the weak convergence for MGCPP:
	\begin{equation} 
	\frac{\vec{S}_{nt} - \Tilde{a^*} \vec{N}_{nt}}{\sqrt{n}} \longrightarrow_{n\rightarrow \infty} \Tilde{\sigma}^* \Lambda^{1/2} \Vec{W}(t), \, for\, all\, t>0.
	\end{equation}   	
\end{proof}

Next, we consider a simple special case. Let $X_{i,k}$ be a Markov chain with two dependent states $(+\delta,-\delta)$ and the ergodic probabilities $\left(\pi_{i}^{*}, 1-\pi_{i}^{*}\right)$. In the limit order market, the $\delta$ is the fixed tick size and the $d$-dimensional point process $\Vec{N}_{nt}$ represents the order flow for $d$ stocks. Here, we set $a_i(x) = x$ in the equation \ref{priceprocess}. In this way, we can derive the corresponding limit theorems for the $d$-dimensional price process $\Vec{S}_{nt}$.

\begin{corollary} \label{fclt1c} (FCLT \RomanNumeralCaps{1} two-state MGCPP: Stochastic Centralization). 
	
	\begin{equation} \label{TFCLTc}
	\frac{\vec{S}_{nt} - \Tilde{a^*} \vec{N}_{nt}}{\sqrt{n}} \longrightarrow_{n \rightarrow \infty} \Tilde{\sigma}^* \Lambda^{1/2} \Vec{W}(t), \, for\, all\, t>0 ,
	\end{equation}  
	where $\Vec{W}(t)$ is a standard $d$-dimensional Brownian motion, ${\Lambda}$ is a diagonal matrix such that $\Lambda = \normalfont{diag} (\bar{\lambda}_1,\, \bar{\lambda}_2,\, \bar{\lambda}_3, \,\cdots, \bar{\lambda}_d) $, $\Tilde{a^*}$ and $\Tilde{\sigma}^*$ are diagonal matrices defined as
	
	\[
	\Tilde{a^*} =
	\begin{bmatrix}
	a^*_1 & \cdots & 0\\
	\vdots& \ddots & \vdots\\
	0 & \cdots & a^*_d
	\end{bmatrix}, \,   \vec{N}_{nt} =
	\begin{bmatrix}
	N_{1,nt}  \\
	\vdots \\
	N_{d,nt}
	\end{bmatrix}, \,\Tilde{\sigma}^* =
	\begin{bmatrix}
	\sigma^*_1 &\cdots & 0 \\
	\vdots & \ddots & \vdots \\
	0 & \cdots & \sigma^*_d
	\end{bmatrix},
	\]
	where $a_i^* = \delta(2\pi_{i}^* - 1)$, and 
	\begin{equation}\label{sigmastar}
	\sigma_{i}^{* 2}:=4 \delta^{2}\left(\frac{1-p_i^{\prime}+\pi_i^{*}\left(p_i^{\prime}-p_i\right)}{\left(p_i+p_i^{\prime}-2\right)^{2}}-\pi_i^{*}\left(1-\pi_i^{*}\right)\right)
	\end{equation}
	$(p_i,p'_i)$ are transition probabilities of the Markov chain $X_{i,k}$. 	
	
\end{corollary} 

\begin{corollary} \label{lln1c} (LLN for two-state MGCPP). 	
	Let $\Vec{S}_{nt}$ be $d$-dimensional general compound point process with two-state Markov chain $X_{i,k}$, we have
	\[
	\frac{\Vec{S}_{nt}}{n} \rightarrow \Tilde{a^*} \Vec{\bar{\lambda}}t, \,\,\,\,\, a.s.
	\]
	Here, $\Tilde{a^*}$ and $\Vec{\bar{\lambda}}$ are constants defined in corollary \ref{fclt1c}.	
\end{corollary}

\begin{proof}  [Proof of Corollary \ref{fclt1c} and \ref{lln1c}]
	Set Markov chain $X_{i,k}$ with two states $(+\delta,-\delta)$ and $a_i(x) = x$ in theorem \ref{FCLT1} and theorem \ref{lln}, we can derive corollary \ref{fclt1c} and \ref{lln1c} directly.
\end{proof}

\begin{rem}\label{aprem}
	From the FCLT \RomanNumeralCaps{1} of MGCPP, we can derive an approximation for the mid-price $\Vec{S}_{nt}$:
	\begin{equation} \label{apfclt1}
	\vec{S}_{nt}  \sim \Tilde{\sigma}^* \Lambda^{1/2} \Vec{W}(t)\sqrt{n} + \Tilde{a^*} \vec{N}_{nt},
	\end{equation} 
	for all $t>0$ and some lagre enough n. Since $\Vec{S}_{nt}$ is the price process in high-frequency trading, the time is always measured in a very short period (eg, milliseconds). So, even if the window size $nt = 10$ seconds with $t = 0.001$, the $n$ will equal to $10,000$ which is a very large number. In this way, it is reasonable to consider this kind of approximation in the LOB. 
\end{rem}

\begin{rem}
	When $\Vec{N}_t$ is a multivariate Hawkes process, the corresponding FCLTs and LLNs for the $\Vec{S}_{nt}$ were considered in \cite{Guo2020}. When we consider an one-dimensional case, if $N_t$ is a renewal process, the corresponding limit theorems for the semi-Markovian model $S_t$ model were discussed in \cite{Swishchuk20172} and \cite{Swishchuk20173}.
\end{rem}

\subsection{Numerical Examples for FCLT: Stochastic Centralization}

In this Section, we tested the FCLT \RomanNumeralCaps{1} of MGCPP model with the LOBSTER data and compared our results with the result simulated by MGCHP in \cite{Guo2020}. 

\subsubsection{Data Description and Parameter Estimation for MGCPPDO}

The level one LOBSTER data on June 21st, 2012 was considered in this paper. In this data, time is measured in milliseconds and the tick size is one cent which means the corresponding $\delta = 0.005$. We can find the basic data description and check the liquidity from Table \ref{datades}:

\begin{table}[h]
	\caption{Data Description and stock liquidity of Microsoft, Intel, Apple, Amazon, and Google for June 21st, 2012.}
	\label{datades}
	\centering
	\resizebox{\textwidth}{!}{\begin{tabular}{ccccc}
		\toprule
		\textbf{Ticker}	& \textbf{\# of Orders in 1 Day}	& \textbf{Avg \# of Orders / Sec} & \textbf{\# of Price Changes in 1 Day} & \textbf{Avg \# of Price Changes / Sec}\\
		\midrule
		INTC & 404986 & 17.3071 & 3218 & 0.1375\\
		MSFT & 411409 & 5.0640 & 4016 & 0.1716\\
		AAPL & 118497 & 5.0640 & 64351 & 2.7500 \\ 
		AMZN & 57515 & 2.4579 & 27558 & 1.1777 \\ 
		GOOG & 49482 & 2.1146 & 24085 & 1.0293 \\ 
		\bottomrule
	\end{tabular}}
\end{table}

Next, we estimate parameters $\Sigma = \mathrm{diag} (\sigma^2_1,\, \sigma^2_2,\, \sigma^2_3, \,\cdots, \sigma^2_d)$ and $\Vec{\bar{\lambda}} = (\bar{\lambda}_1, \bar{\lambda}_2, \bar{\lambda}_3, \cdots, \bar{\lambda}_d) $ via the LLN and FCLT assumptions of $\Vec{N}_t$. From \ref{as1} and \ref{as2}, when $n$ is large enough, we can derive the approximations:

\begin{equation}  \label{ap1}
\frac{\Vec{N}(nt)}{nt} \sim \Vec{\bar{\lambda}}, \,\,\, t \in [0,1]
\end{equation}
and 
\begin{equation} \label{ap2}
\frac{1}{\sqrt{n}}(\Vec{N}_{nt} - E(\Vec{N}_{nt})) \sim \Sigma^{1/2}\Vec{W}_t, \,\,\, t \in [0,1].
\end{equation} 

Take the expectation for (\ref{ap1}) and variance for (\ref{ap2}), we have
\begin{equation} \label{expnt}
\frac{\mathrm{E}(\Vec{N}(nt))}{nt} \sim \Vec{\bar{\lambda}}, \,\,\, t \in [0,1]
\end{equation}
and 
\begin{equation} 
\frac{1}{nt}(\mathrm{Var}(\Vec{N}_{nt})) \sim \Sigma, \,\,\, t \in [0,1].
\end{equation} 

In this way, we derived the estimated parameters for $5$ in Table \ref{paraorder}. 
\begin{table}[H]
	\caption{Estimated parameters of 5 stocks via the LLN and FCLT assumptions}
	\centering
	\begin{tabular}{cccc}
		\toprule
		\textbf{Ticker} & $\boldsymbol{\sigma}$ & $\boldsymbol{\sigma^2}$ & $\boldsymbol{\bar{\lambda}}$ \\ \midrule
		INTC   & 1.4380   & 2.0680 & 0.1366          \\ 
		MSFT   & 1.1390   & 1.2973 & 0.1729          \\ 
		AAPL   & 7.8981   & 62.38     & 2.2938          \\ 
		AMZN   & 4.3919  & 19.2883 & 1.0374          \\ 
		GOOG   & 4.7747   & 22.7980 & 0.8178          \\ \bottomrule
	\end{tabular}
	\label{paraorder}
\end{table}

Next, we estimated parameters for the Markov chain by applying the two-state MGCPP model in corollary \ref{fclt1c}. The transition matrix $P$ of two dependent state Markov chain $X_k$ is denoted as
\[P=\left[ \begin{array}{cc}{p_{u u}} & {1-p_{u u}} \\ {1-p_{d d}} & {p_{d d}}\end{array}\right].\]
We calculated frequency in our data to estimated the $p_{uu}$ and $p_{dd}$ in $P$ by 
\begin{align*}
p_{uu}&=\frac{q_{uu}}{q_{uu}+q_{ud}},\\
p_{dd}&=\frac{q_{dd}}{q_{dd}+q_{du}},
\end{align*}
where $q_{uu}$, $q_{dd}$, $q_{ud}$, and $q_{du}$ are the number of price goes up twice, goes down twice, goes up and then down, goes down and then up, respectively. And the result is in Table \ref{tran}:
\begin{table}[H]
	\caption{Transition matrix and constant parameters for two-state MGCPP. $\alpha^*$ and $\sigma^*$ were calculated by equation (\ref{sigmastar}). }
	\centering
	\begin{tabular}{ccccc}
		\toprule
		\textbf{Ticker} & $\boldsymbol{p_{uu}}$ & $\boldsymbol{p_{dd}}$ & $\boldsymbol{\sigma^*}$ & $\boldsymbol{a^*}$       \\ \midrule
		INTC   & 0.5373   & 0.5814   & 0.0057     & -2.5023$\times 10^{-4}$ \\ 
		MSFT   & 0.5711   & 0.6044   & 0.0060     & -2.0145$\times 10^{-4}$ \\ 
		AAPL   & 0.4954   & 0.4955   & 0.0050     & -2.1529$\times 10^{-7}$ \\ 
		AMZN   & 0.4511   & 0.4590   & 0.0046     & -3.6077$\times 10^{-5}$ \\ 
		GOOG   & 0.4536   & 0.4886   & 0.0047     & -1.6584$\times 10^{-4}$  \\ \bottomrule
	\end{tabular}
	\label{tran}
\end{table}

\subsubsection{Comparison with multivariate general compound Hawkes process with two dependent orders}

In this Section, we compared the simulation results of MGCPP with the multivariate general compound Hawkes process (MGCHP) model to show that the simple generalized model can also reach a good accuracy as the MGCHP who has a sophisticated intensity function. In \cite{Guo2020}, they simulated the MGCHP with two dependent states for Microsoft and Intel's data. So here we also conduct simulations for Microsoft and Intel's data with the two-state MGCPP, which means the Markov chain has two dependent states $(+\delta, -\delta)$.  

We tested the MGCPP model by comparing the standard deviation for the left hand side and right hand side in the FCLT:
\begin{equation*}
\frac{\vec{S}_{nt} - \Tilde{N} _{nt}\vec{a^*}}{\sqrt{n}} \longrightarrow_{n \rightarrow \infty} \Tilde{\sigma}^* \Lambda^{1/2} \vec{W}(t).
\end{equation*}
That is to say, we first cut our data into disjoint windows of size $nt$, specifically $[i nt,(i+1) nt] \text { with } t=0.001$ and by setting the left bound as our starting time we can calculate:
\[\vec{S}_{i}^{*}=\vec{S}_{(i+1) n t}-\vec{S}_{i n t}-(\Tilde{N}((i+1) n t)-\Tilde{N}(i n t)) \vec{a^*},\]
and the equation for standard deviation is given by

\begin{equation}\label{stdsqrt}
\operatorname{std}\left\{\vec{S}^{*}\right\} \approx \sqrt{n}  \Tilde{\sigma}^* \Lambda^{1/2} \sqrt{\vec{t}}.
\end{equation}

The Figure \ref{std} gives a standard deviation comparison of MGCPP, MGCHP, and the raw data for 2 stocks in different
window sizes from 0.1 second to 12 seconds in steps of 0.1 second. First, we could find the MGCPP parameters make the standard deviation of LHS very similar to the RHS for each stocks when $n$ is large. So, generally speaking, we can say our MGCPP model fits the data well. Second, the MGCPP curve is very close to the MGCHP curve or we could say the simulation results via Intel and Microsoft stocks data are nearly same. It shows that even we don't have a sophisticated intensity function as the Hawkes process, we still can reach a relative good result with a simple point process model. This can help us deal with the computing efficiency problem when using the MGCHP model. We'll give more quantitative error analysis later.

\begin{figure}[H]
	\centering	
	\begin{minipage}{0.47\textwidth}
		\includegraphics[width=\linewidth]{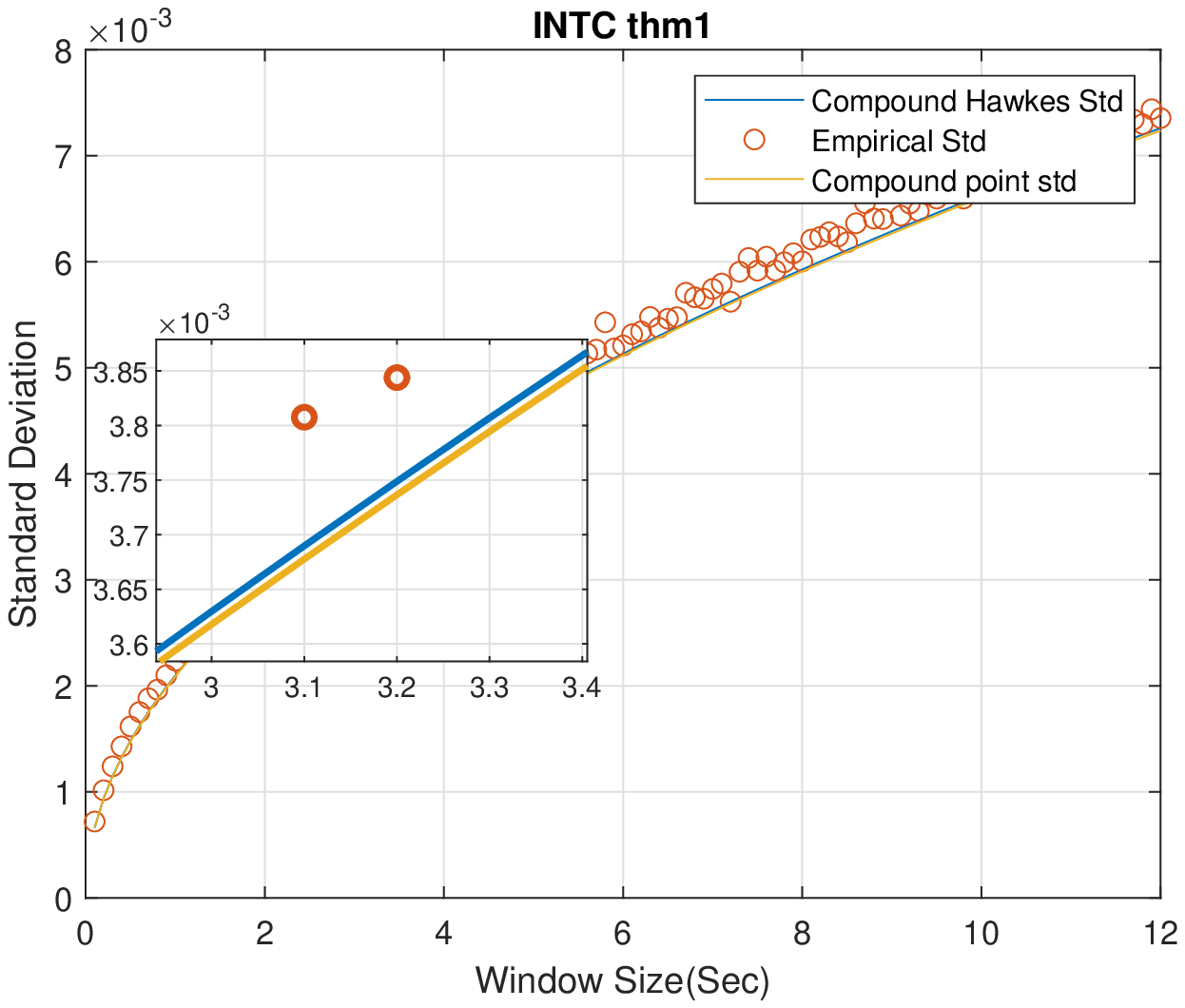}
		\label{fig:1}
	\end{minipage}
	\hspace{3mm} 
	\begin{minipage}{0.47\textwidth}
		\includegraphics[width=\linewidth]{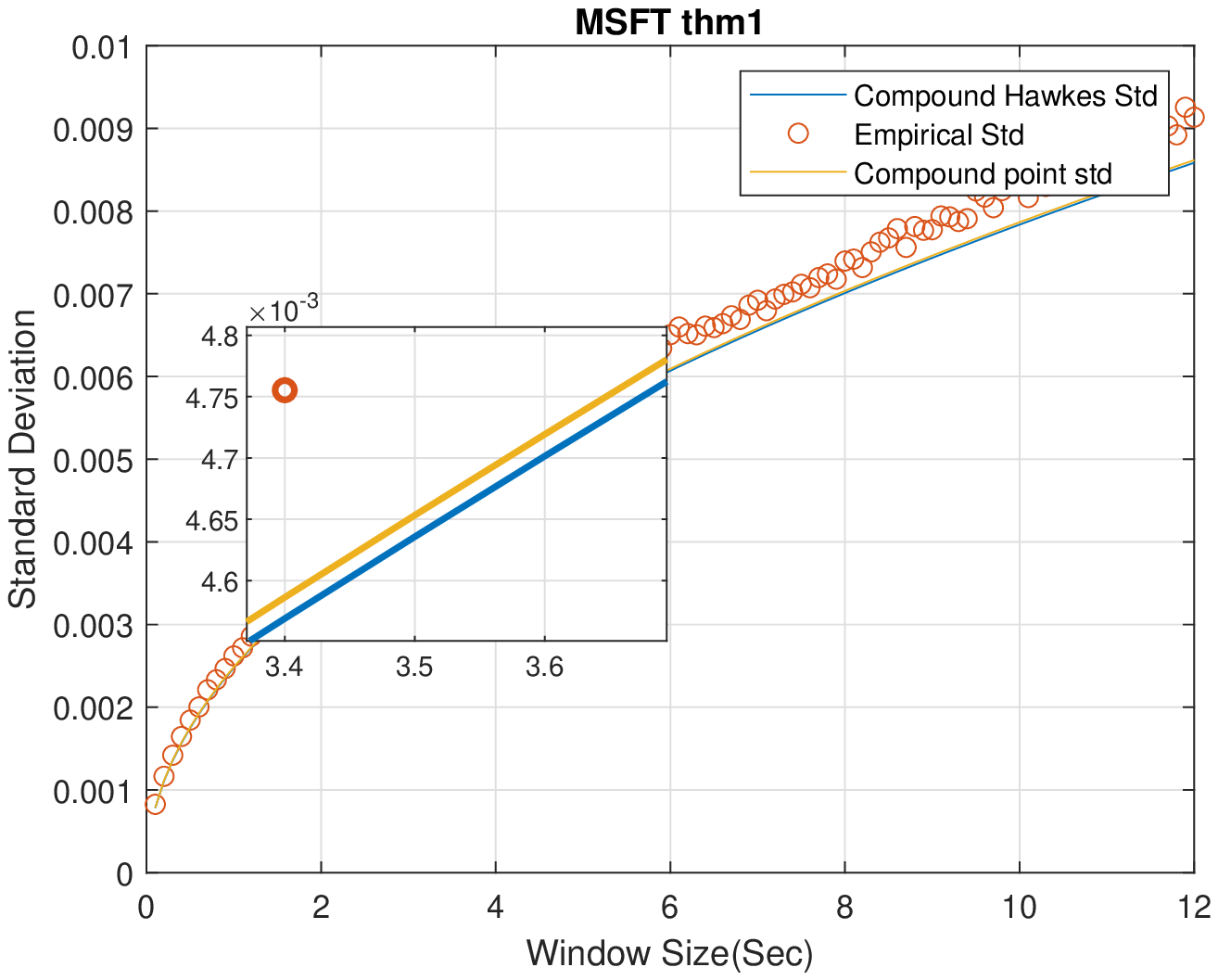}
		\label{fig:2}
	\end{minipage}
	\caption{Standard deviation comparisons for 2 stocks by FCLT \RomanNumeralCaps{1} for MGCHP and MGCPP}
	\label{std}
\end{figure}

\begin{rem}
	Since the number of windows decreases as the window size $nt$ increases, we can find that the spread of data increases when the window size increases in Figure \ref{std}. For example, when we consider $nt = 0.1$ second, the number of windows is 234,000. However, a 12-second window size yields 1,950 windows which will lead the standard deviation increases.
\end{rem}

Intuitively, the Figure \ref{std} shows that the standard deviation of MGCHP and MGCPP are very close and both of them fit the real standard deviation very well. Next, we analyze MGCHP and MGCPP models quantitatively. 

We computed the mean square error (MSE) of the real standard deviation and theoretical standard deviations in Table \ref{MSE}. As can be seen from the Table \ref{MSE}, MGCHP model performs better than the MGCPP model with both Intel and Microsoft data. For Intel stock data, the MSE of MGCHP is $17\%$ better than MGCPP and nearly $10\%$ better than MGCPP model with the Microsoft stock data. However, when we compare the order of magnitude of the MSE ($-8$) with the real standard deviation ($-2$ and $-3$), we still can conclude that MGCPP is good enough for the mid-price modeling task.

\begin{table}[H]
	\caption{The MSE of the real standard deviation and theoretical standard deviations from MGCHP and MGCPP. }
	\centering
	\begin{tabular}{ccc}
		\toprule
		\textbf{Ticker} & \textbf{MGCHP MSE}               & \textbf{MGCPP MSE }              \\ \midrule
		INTC   & $3.4039 \times 10^{-8}$ & $3.9858 \times 10^{-8}$ \\ 
		MSFT   & $9.6454 \times 10^{-8}$ & $8.6189 \times 10^{-8}$ \\ \bottomrule
	\end{tabular}
	\label{MSE}
\end{table}

Recall the equation (\ref{stdsqrt}), we can find the standard deviation and the square root of time step have a linear relationship. So, we can fit the real standard deviation data with the square root curve by using the least-square regression. And then, we can compare the coefficients from the least-square regression and two stochastic models. 

\begin{table}[H]
	\centering
	\caption{Coefficients calculated by MGCHP and MGCPP models. }
	\resizebox{\textwidth}{!}{\begin{tabular}{cccccc}
		\toprule
		\textbf{Ticker} & \textbf{MGCHP Coefficient} & \textbf{MGCPP Coefficient}  & \textbf{Regression Coefficient}  & \textbf{MGCHP \% Error}   & \textbf{MGCPP \% Error }  \\ \midrule
		INTC   & 0.002086          & 0.002089          & 0.002162               & $3.515\%$              & $3.377\%$              \\ 
		MSFT   & 0.002494          & 0.002487          & 0.002609               & $4.408\%$              & $4.676\%$              \\ \bottomrule
	\end{tabular}}
	\label{coeff}
\end{table}

From the Table \ref{coeff}, we can find that the percentage error of both two stochastic models are all smaller than $5\%$ and there is no significant difference between the MGCPP coefficient and the MGCHP coefficient.

\subsubsection{MGCPP with n-state Dependent Orders}

We will give more simulation examples by using the Google, Apple, and Amazon data with the MGCPP model with n-state dependent orders in this Section. Thanks to \cite{Swishchuk2020}, we can conclude that the accuracy of the general compound Hawkes process model increases when the number of states increases. And for Google, Apple, and Amazon in LOBSTER data set, the best number of states is $4$ to $7$. In the previous Section, we also showed that the simulation results of MGCPP is nearly same as the MGCHP. So, it's reasonable to consider a MGCPP model with $7$-state Markov chain here.

\begin{figure}[H]
	\centering	
	\begin{minipage}{0.42\textwidth}
		\includegraphics[width=\linewidth]{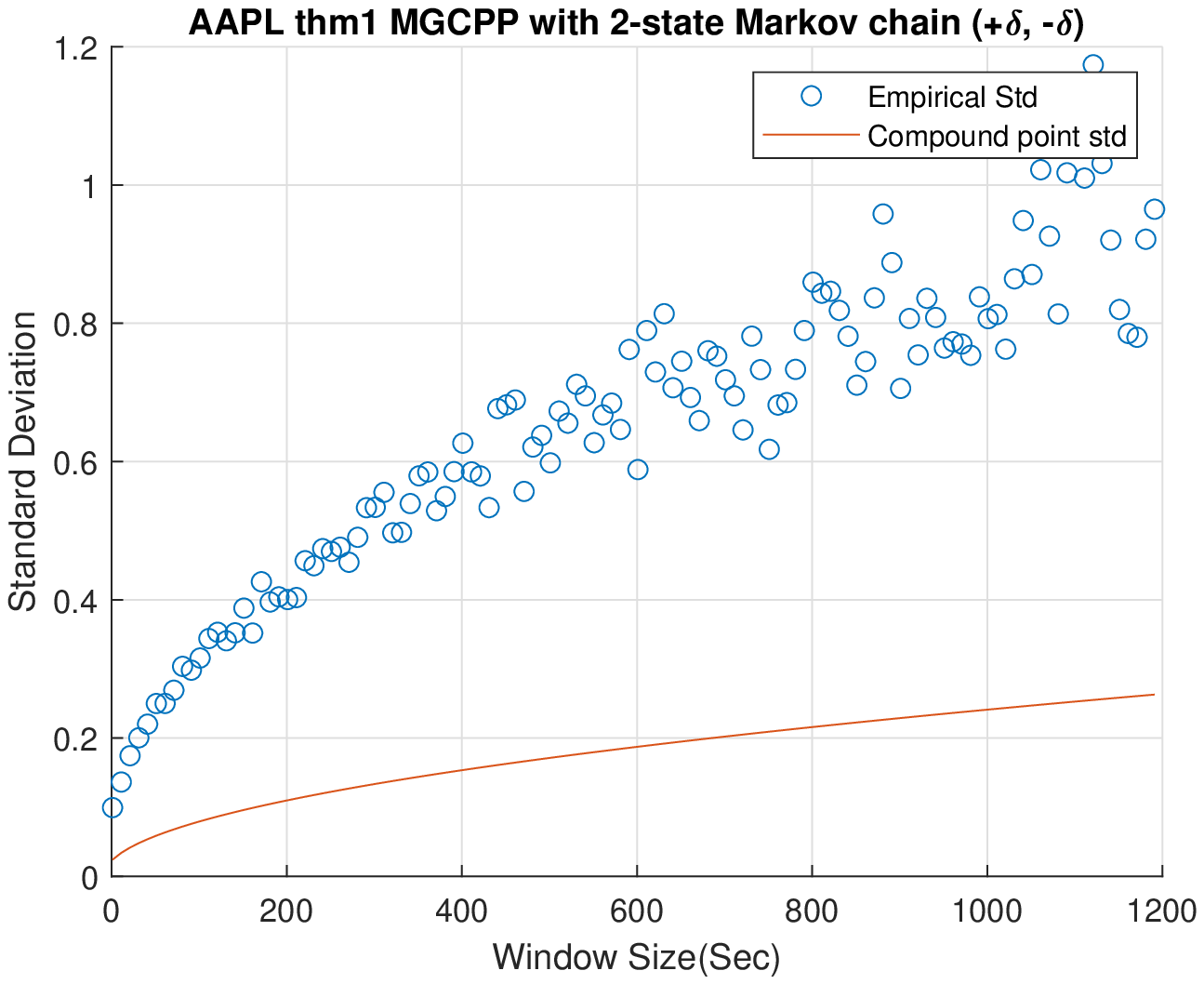}
		\label{aapl2}
	\end{minipage}
	\hspace{3mm} 
	\begin{minipage}{0.42\textwidth}
		\includegraphics[width=\linewidth]{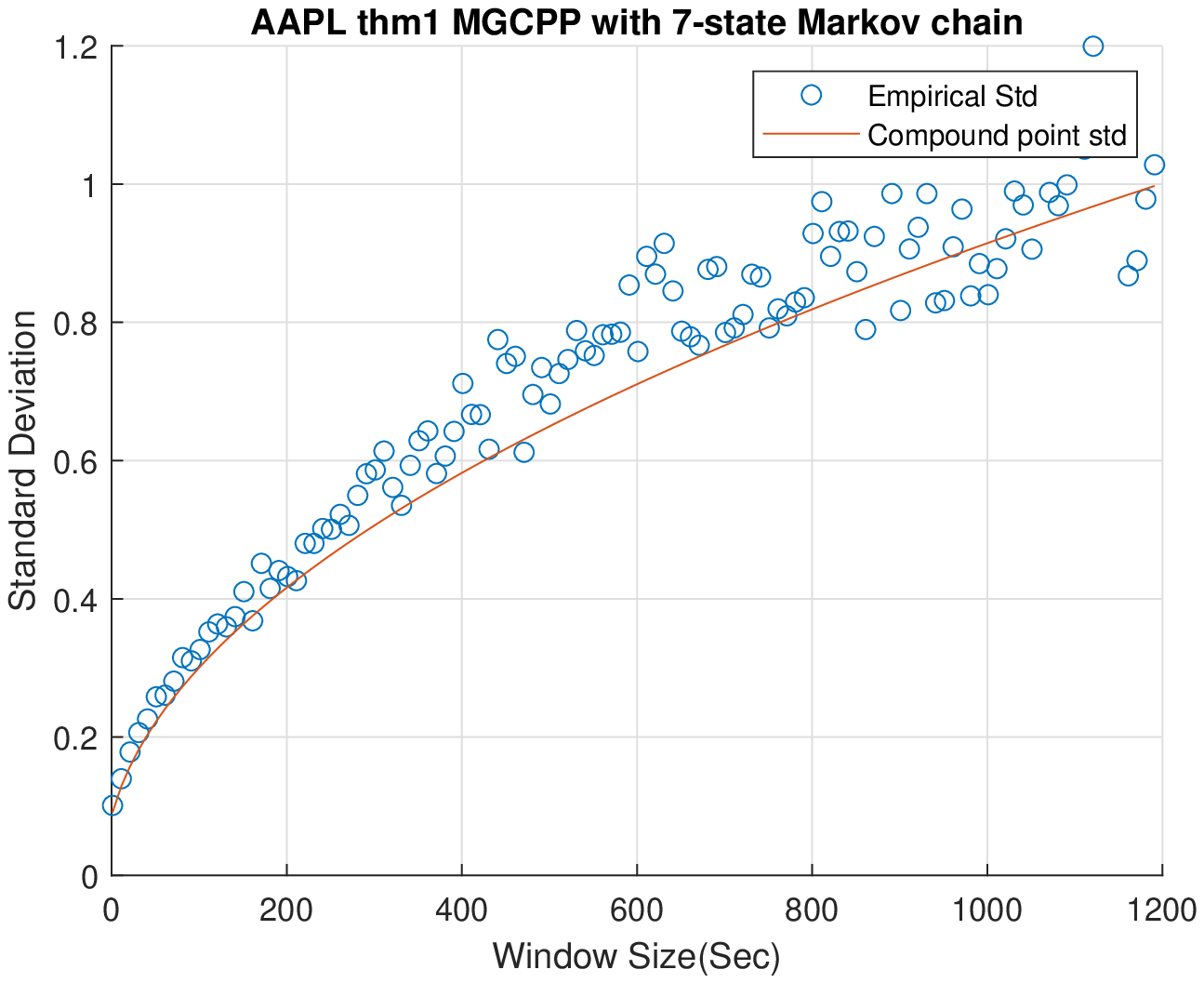}
		\label{aapl7}
	\end{minipage}
	\caption{Standard deviation comparisons for MGCPP with 2-state Markov chain and 7-state Markov chain simulated by Apple's stock data}
	\label{stdaapl}
\end{figure}

\begin{figure}[H]
	\centering	
	\begin{minipage}{0.42\textwidth}
		\includegraphics[width=\linewidth]{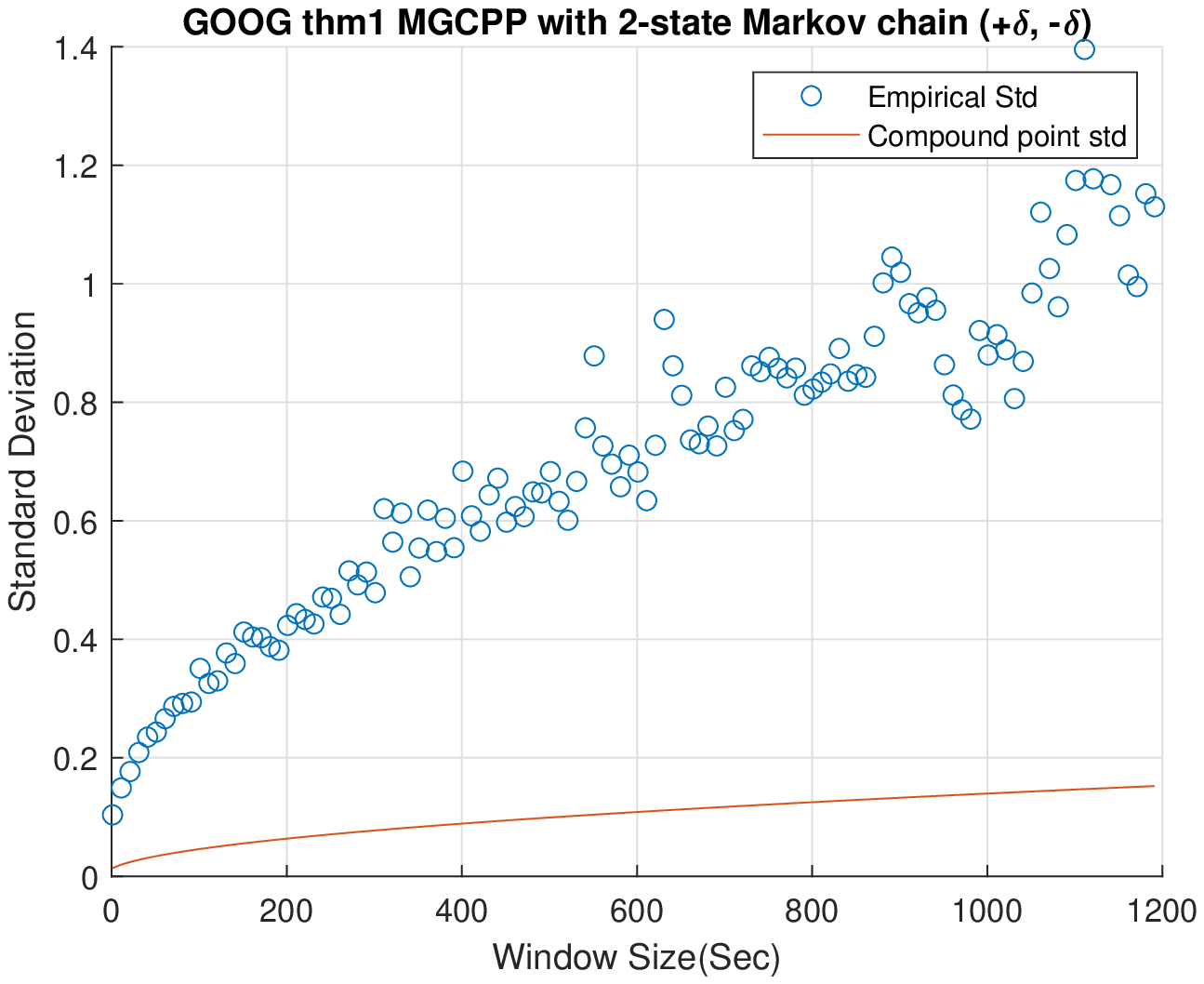}
		\label{goog2}
	\end{minipage}
	\hspace{3mm} 
	\begin{minipage}{0.42\textwidth}
		\includegraphics[width=\linewidth]{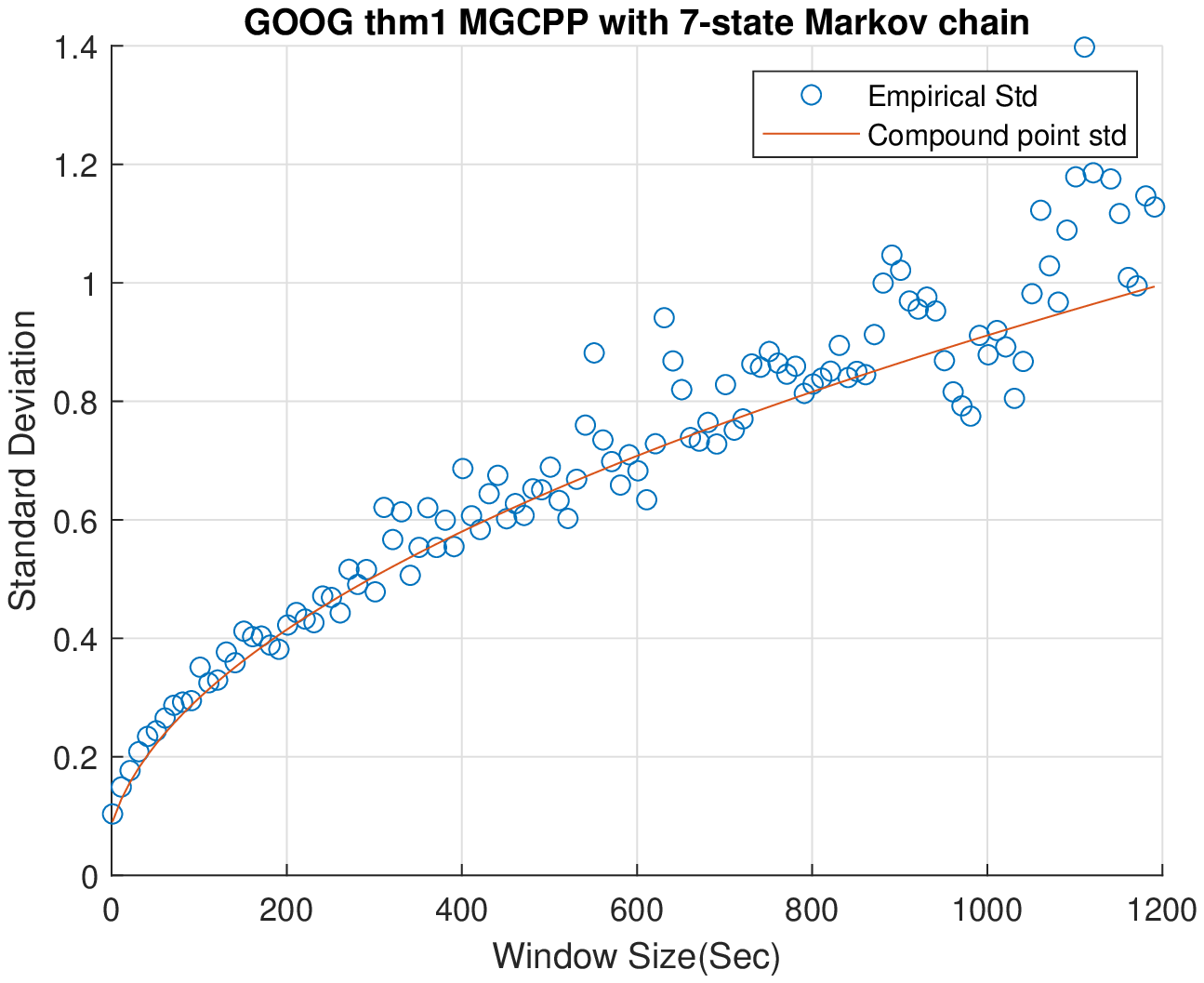}
		\label{goog7}
	\end{minipage}
	\caption{Standard deviation comparisons for MGCPP with 2-state Markov chain and 7-state Markov chain simulated by Google's stock data}
	\label{stdgoog}
\end{figure}

\begin{figure}[H]
	\centering	
	\begin{minipage}{0.42\textwidth}
		\includegraphics[width=\linewidth]{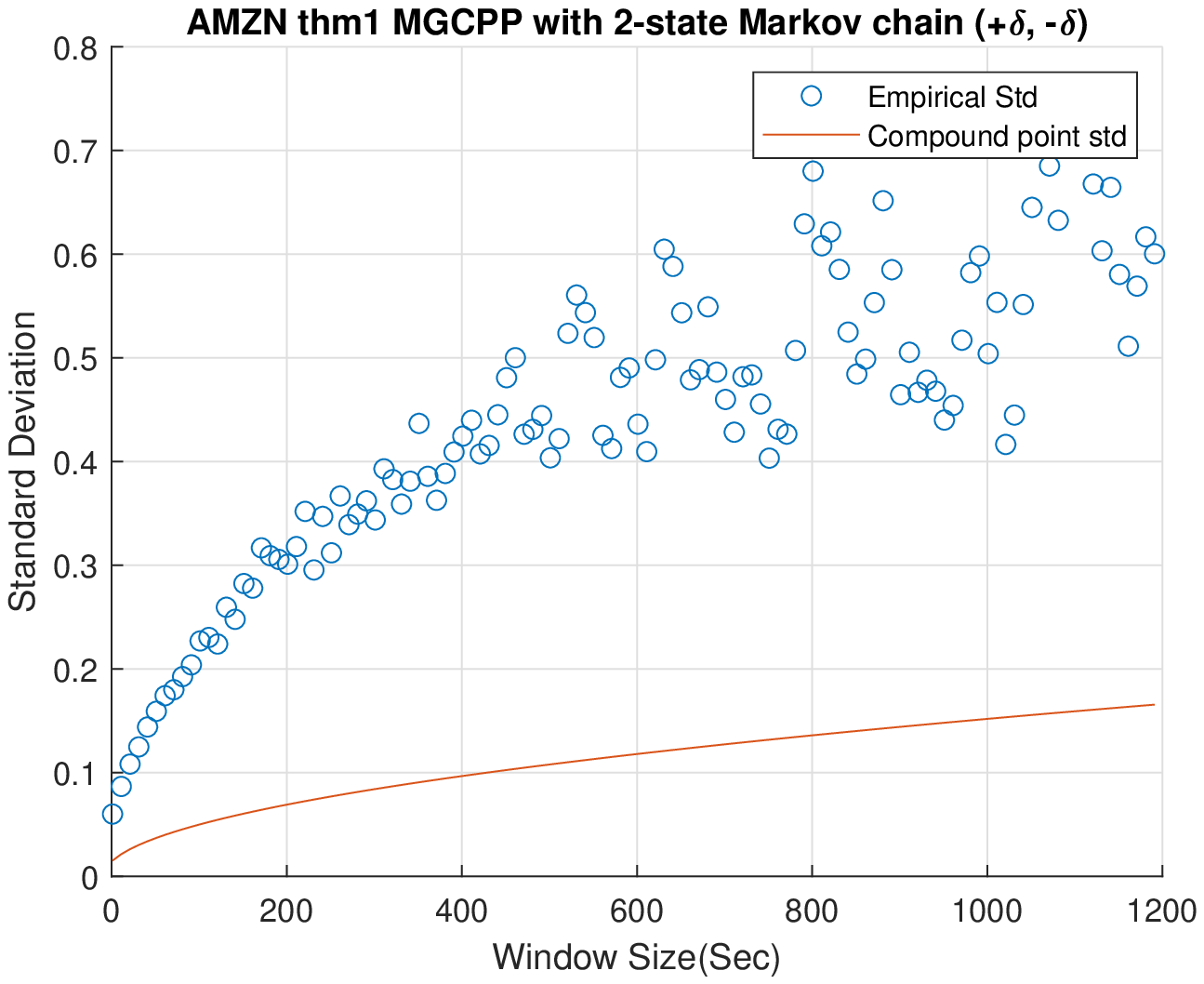}
		\label{amzn2}
	\end{minipage}
	\hspace{3mm} 
	\begin{minipage}{0.42\textwidth}
		\includegraphics[width=\linewidth]{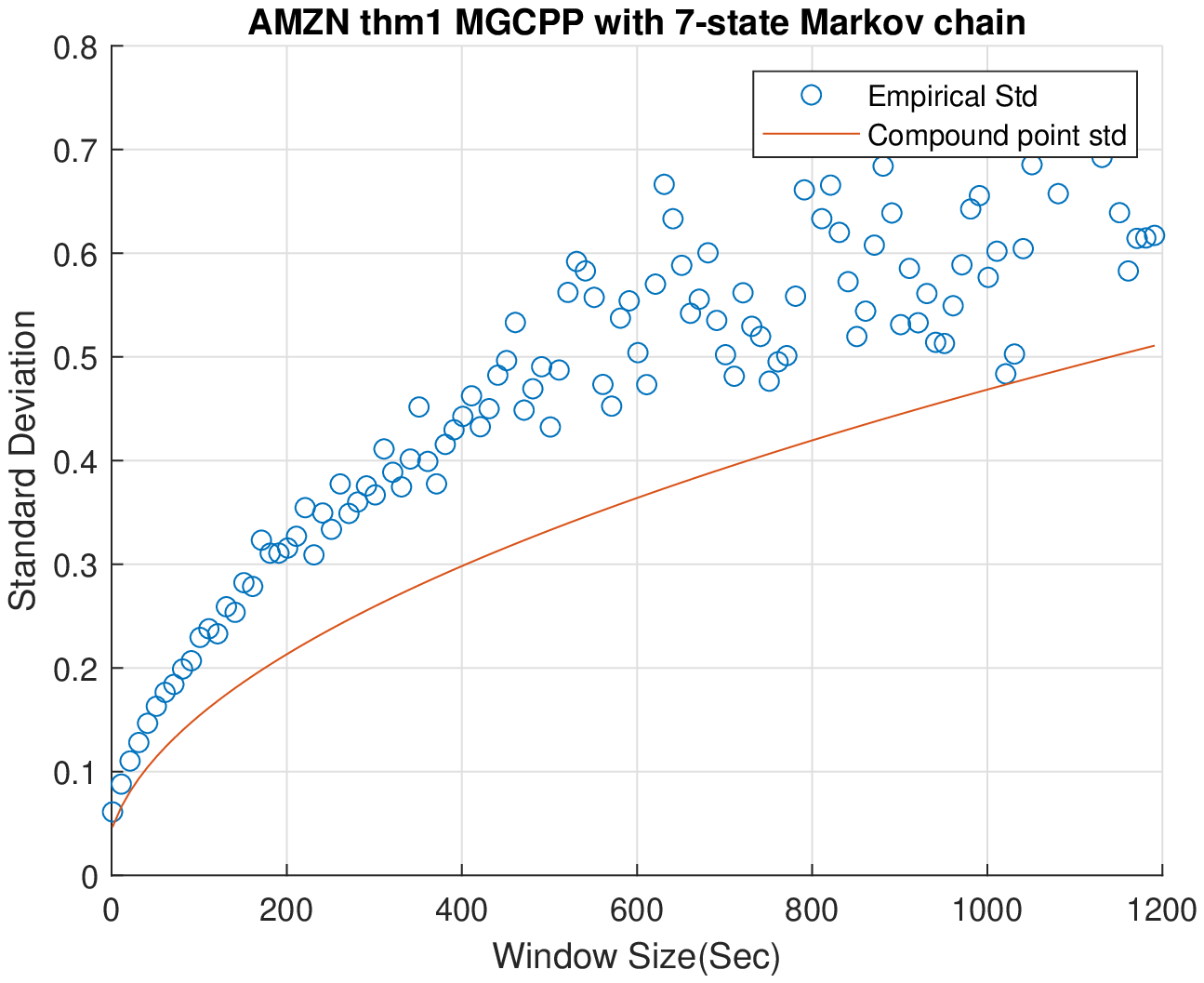}
		\label{amzn7}
	\end{minipage}
	\caption{Standard deviation comparisons for MGCPP with 2-state Markov chain and 7-state Markov chain simulated by Amazon's stock data}
	\label{stdamzn}
\end{figure}

Figure \ref{stdaapl}, \ref{stdgoog}, and \ref{stdamzn} give standard deviation comparisons for MGCPP with 2-state Markov chain and 7-state Markov chain simulated by different tickers' data. Since the 2-state simulation results here are not as good as the results simulated by Intel's and Microsoft's data, we take bigger time steps and window sizes (from 10 seconds to 20 minus with 10 seconds time step) to capture more dynamics. From the figures we can find that the 7-state model has a significant improvement than the 2-state model. 7-state curves for AAPL and GOOG are very close to the real standard deviation, although the theoretical curve of AMZN is underestimated even with the 7-state model.  

\begin{table}[H]
	\centering
	\caption{The MSE and coefficients computed by MGCPP with 2-state and 7-state Markov chain for different tickers. The regression coefficients were derived by fitting the real standard deviations with square root curve. And MGCPP coefficients were computed by equation (\ref{stdsqrt}).}
	\begin{tabular}{ccccc}
		\toprule
		\textbf{Ticker}       & \textbf{MSE}     & \textbf{Regression Ceofficient}  & \textbf{MGCPP Ceofficient}  & \textbf{Percentage Error}  \\ \midrule
		AAPL 2-state & 0.2467 & 0.0278                 & 0.0076            & $72.66\%$       \\ 
		AAPL 7-state & 0.0064 & 0.0311                 & 0.0288            & $7.40\%$         \\ 
		GOOG 2-state & 0.4161 & 0.0307                 & 0.0044            & $85.67\%$       \\ 
		GOOG 7-state & 0.0081 & 0.0307                 & 0.0287            & $6.51\%$       \\ 
		AMZN 2-state & 0.1233 & 0.0189                 & 0.0048            & $74.60\%$       \\ 
		AMZN 7-state & 0.0225 & 0.0205                 & 0.0147            & $28.29\%$        \\ \bottomrule
	\end{tabular}
	\label{coeffnstate}
\end{table}
The Table \ref{coeffnstate} lists the MSE and coefficients of the 2-state and 7-state models with different tickers. We can find the improvement of 7-state model quantitatively from the Table. The results of AAPL and GOOG are good enough for the mid-price modeling. As for AMZN, although we derive a remarkable improvement from 2-state model ($74.60\%$ error) to 7-state model ($28.29\%$ error), we cannot make the error smaller than $5\%$ or $10\%$. This is to say, MGCPP model may not be able to capture the full dynamics for AMZN data, but it still can be a strong candidate for modeling the mid-price, which is consistent with the conclusion of MGCHP model in \cite{Swishchuk2020}.  

In general, we can conclude that: as a generalization of MGCHP, the MGCPP model also has a very good performance in mid-price dynamics modeling. If we consider the MGCPP with higher states Markov chain, we will derive a better result.

\begin{rem}
	The MGCPP is not only a generalization of MGCHP, but also a generalization for all multivariate compound models whose point processes $\Vec{N_t}$ satisfy the assumptions \ref{as1} and \ref{as2}. The reason we use Hawkes process for comparison is we want to take the advantage of numerical examples in references. 
\end{rem}

\section{Diffusion limit for the MGCPP: Deterministic Centralization}\label{sec4}

We proved a LLN and FCLT for the MGCPP in the previous Section. And the limit theorems provide us an approximation for the mid-price modeling in the LOB. Recall the approximation in Remark \ref{aprem}, we have 
\begin{equation} \label{fcltap1}
\vec{S}_{nt}  \sim \Tilde{\sigma}^* \Lambda^{1/2} \Vec{W}(t)\sqrt{n} + \Tilde{a^*} \vec{N}_{nt},
\end{equation}
where the $\vec{S}_{nt}$ is the price process and $\vec{N}_{nt}$ is the order flow. However, in the real-world problems, equation (\ref{fcltap1}) cannot help us with the forecasting task directly because we couldn't have the order flow $\vec{N}_{nt}$ in advance. This motivates us to consider a FCLT \RomanNumeralCaps{2} for the MGCPP in this Section.

\subsection{FCLT for MGCPP: Deterministic Centralization}
\begin{theorem} \label{fclt2t} (FCLT \RomanNumeralCaps{2}: Deterministic Centralization). Let $X_{i,k}, \, i=1,2,\cdots,d$ be independent ergodic Markov chains with $n$ states $\{1,2,\cdots,n\}$ and with ergodic probabilities $\left(\pi_{i,1}^{*}, \pi_{i,2}^{*}, \ldots, \pi_{i,n}^{*}\right)$. Let $\vec{S}_{nt}$ be $d$-dimensional compound point process, we have 
	\begin{equation} \label{FCLT2}
	\frac{\vec{S}_{nt} - \Tilde{a^*}  E(\vec{N}_{nt})}{\sqrt{n}} \longrightarrow \Tilde{\sigma}^* \Lambda^{1/2} \Vec{W}_1(t) + \Tilde{a^*} \Sigma^{1/2}\Vec{W}_2(t), \, for\, all\, t>0 
	\end{equation}   
	as $n \rightarrow \infty$, where $\Vec{W}_1(t)$ and $\Vec{W}_2(t)$ are independent standard $d$-dimensional Brownian motions. Parameters $\Tilde{\sigma}^*$, $\Tilde{a^*}$, $\Lambda$, and $\Sigma$ are defined in Theorem \ref{FCLT1}.
\end{theorem}

\begin{proof}[Proof of Theorem \ref{fclt2t}]
	Recall the FCLT for MPP (assumption \ref{as2}), we have
	\begin{equation}\label{PPfclt}
	\bigg(  \frac{1}{\sqrt{n}} \vec{N}_{nt} -\frac{1}{\sqrt{n}} E(\Vec{N}_{nt}) \bigg) \longrightarrow \mathbf{\Sigma}^{1/2}\Vec{W}_t
	\end{equation}
	in law for the Skorokhod topology, as $n \rightarrow \infty$. And from theorem \ref{FCLT1}, we have the FCLT  for MGCPP
	\begin{equation} \label{fclt1r}
	\frac{\vec{S}_{nt} - \Tilde{a^*} \vec{N}_{nt}}{\sqrt{n}} \longrightarrow \Tilde{\sigma}^* \Lambda^{1/2} \Vec{W}_t, \, for\, all\, t>0 
	\end{equation}   
	as $n \rightarrow \infty$ in the weak law of Skorokhod topology. Here, we assume two multivariate Brownian motions in (\ref{PPfclt}) and (\ref{fclt1r}) are mutually independent and we refer them $\vec{W}_2(t)$ and $\vec{W}_1(t)$. Next, consider
	\begin{equation}
	\frac{\vec{S}_{nt}}{\sqrt{n}} -  \frac{\Tilde{a}^*   E(\Vec{N}_{nt})}{\sqrt{n}}  = \frac{\vec{S}_{nt} - \Tilde{a}^* \vec{N}_{nt}}{\sqrt{n}}  +  \Tilde{a}^* \bigg(  \frac{1}{\sqrt{n}} \vec{N}_{nt} -\frac{1}{\sqrt{n}}E(\Vec{N}_{nt}) \bigg).
	\end{equation}   
	With (\ref{PPfclt}) and (\ref{fclt1r}) we can derive
	\begin{equation}
	\frac{\vec{S}_{nt} - \Tilde{a}^* \vec{N}_{nt}}{\sqrt{n}}  +  \Tilde{a}^* \bigg(  \frac{1}{\sqrt{n}} \vec{N}_{nt} -\frac{1}{\sqrt{n}}E(\Vec{N}_{nt}) \bigg) \longrightarrow \Tilde{\sigma}^* \Lambda^{1/2} \Vec{W}_1(t) + \Tilde{a^*} \Sigma^{1/2}\Vec{W}_2(t)
	\end{equation} 
	as $n \rightarrow \infty$ which gives (\ref{FCLT2}).
\end{proof}

\begin{rem}
	We can also consider a special case as the FCLT \RomanNumeralCaps{1}. Let $X_{i,k}$ be a Markov chain with two dependent states $(+\delta,-\delta)$ and the ergodic probabilities are $\left(\pi_{i}^{*}, 1-\pi_{i}^{*}\right)$. Set $a_i(x) = x$ in the definition \ref{priceprocess}. Then, we can derive a similar result for FCLT \RomanNumeralCaps{2}. Parameters $\Tilde{a}^*$ and $\Tilde{\sigma}^*$ can be computed by equation (\ref{sigmastar}).
\end{rem}

\begin{rem} \label{rem3}
	For the FCLT \RomanNumeralCaps{2}, we can also consider a similar approximation as the FCLT \RomanNumeralCaps{1}. For some large enough $n$, we have
	\begin{equation} \label{FCLTap2}
	\vec{S}_{nt}  \sim \sqrt{n} \Tilde{\sigma}^* \Lambda^{1/2} \Vec{W}_1(t) + \sqrt{n} \Tilde{a^*} \Sigma^{1/2}\Vec{W}_2(t) + \Tilde{a^*}  E(\vec{N}_{nt}), \, \text{for all $t>0$}. 
	\end{equation}  
	To deal with the $E(\vec{N}_{nt})$ term, we consider the approximation derived from assumption \ref{as1} in equation (\ref{expnt}):
	\begin{equation}
	E(\Vec{N}(nt)) \sim nt\Vec{\bar{\lambda}}.
	\end{equation}	
	Rewrite equation (\ref{FCLTap2}), we have the new approximation
	\begin{equation} \label{FCLTap3}
	\vec{S}_{nt}  \sim \sqrt{n} \Tilde{\sigma}^* \Lambda^{1/2} \Vec{W}_1(t) + \sqrt{n} \Tilde{a^*} \Sigma^{1/2}\Vec{W}_2(t) + \Tilde{a^*}  nt\Vec{\bar{\lambda}}.
	\end{equation}	
\end{rem}

\subsection{Numerical Examples for FCLT: Deterministic Centralization}

In this Section, we applied the LOBSTER data to test the FCLT \RomanNumeralCaps{2}. According to the numerical examples of FCLT \RomanNumeralCaps{1}, we consider the standard deviation of the approximation in Remark \ref{rem3}, namely 

\begin{equation} 
\operatorname{std}\left\{\vec{S}_{(i+1) n t}-\vec{S}_{i n t}\right\}  \approx \sqrt{ (\Tilde{\sigma}^*)^2 \Lambda n\vec{t}  + (\Tilde{a^*})^2 \Sigma n\vec{t}}.
\end{equation}

The comparisons of real standard deviation and theoretical standard deviation can be found in Figure \ref{stdfclt2}. Since results of INTC and MSFT are good enough with the $2$-state Markov chain $(+\delta, -\delta)$ in FCLT \RomanNumeralCaps{1}, we also applied $2$-state Markov chain for INTC and MSFT here. As for AAPL, GOOG, and AMZN, we used the MGCPP model with $7$-state Markov chain. Window sizes here start from 1 second and increase to 20 minutes in time steps of 10 seconds. As can be seen in Figure \ref{stdfclt2}, the results for FCLT \RomanNumeralCaps{2} are as good as the FCLT \RomanNumeralCaps{1} results in Figure \ref{std}, \ref{stdaapl}, \ref{stdgoog}, and \ref{stdamzn}. We also computed the MSE and coefficients in Table \ref{coeff2}.

\begin{figure}[H]
	\centering
	\includegraphics[width=.4\textwidth]{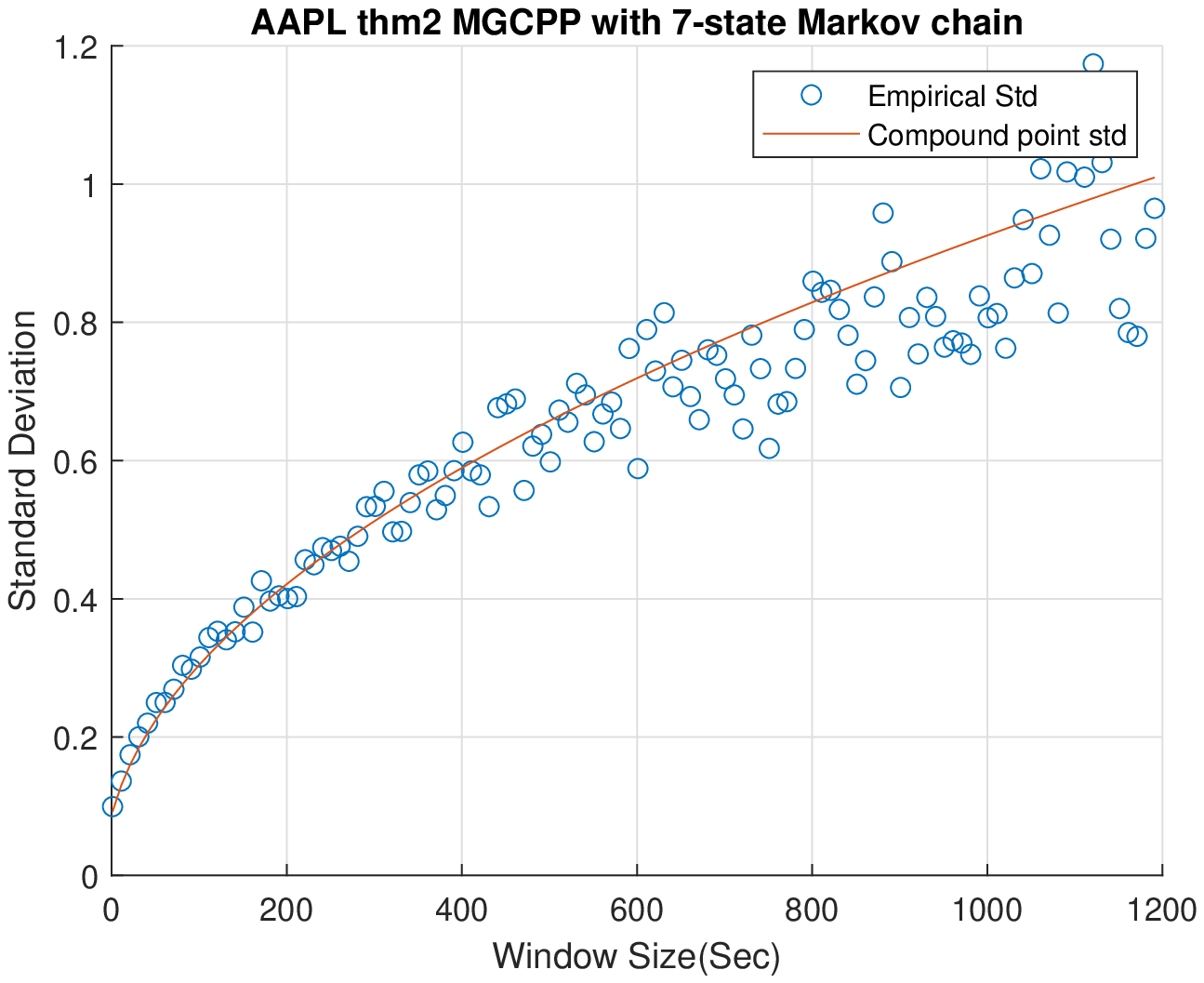}\quad
	\includegraphics[width=.4\textwidth]{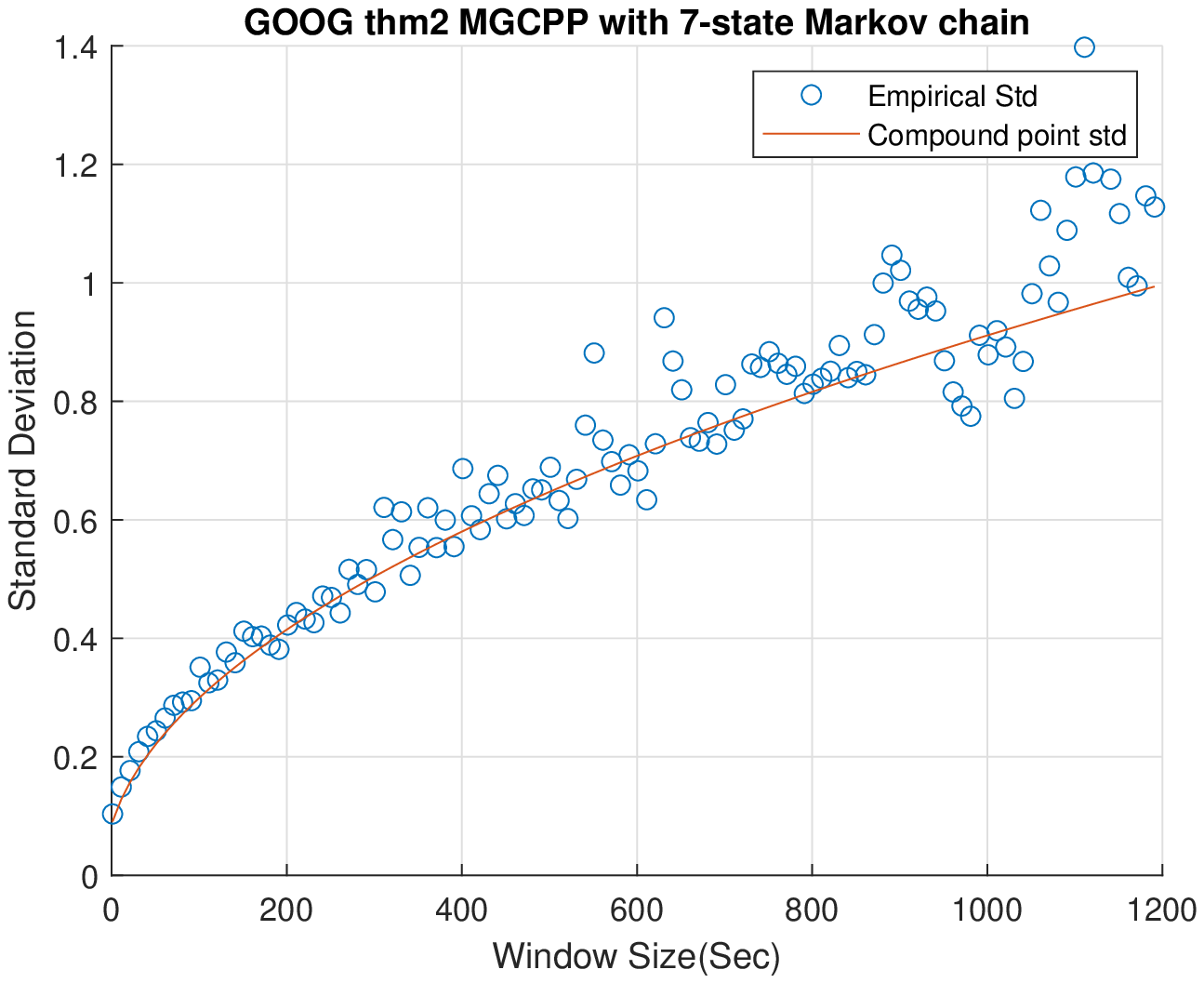}
	\medskip
	\includegraphics[width=.4\textwidth]{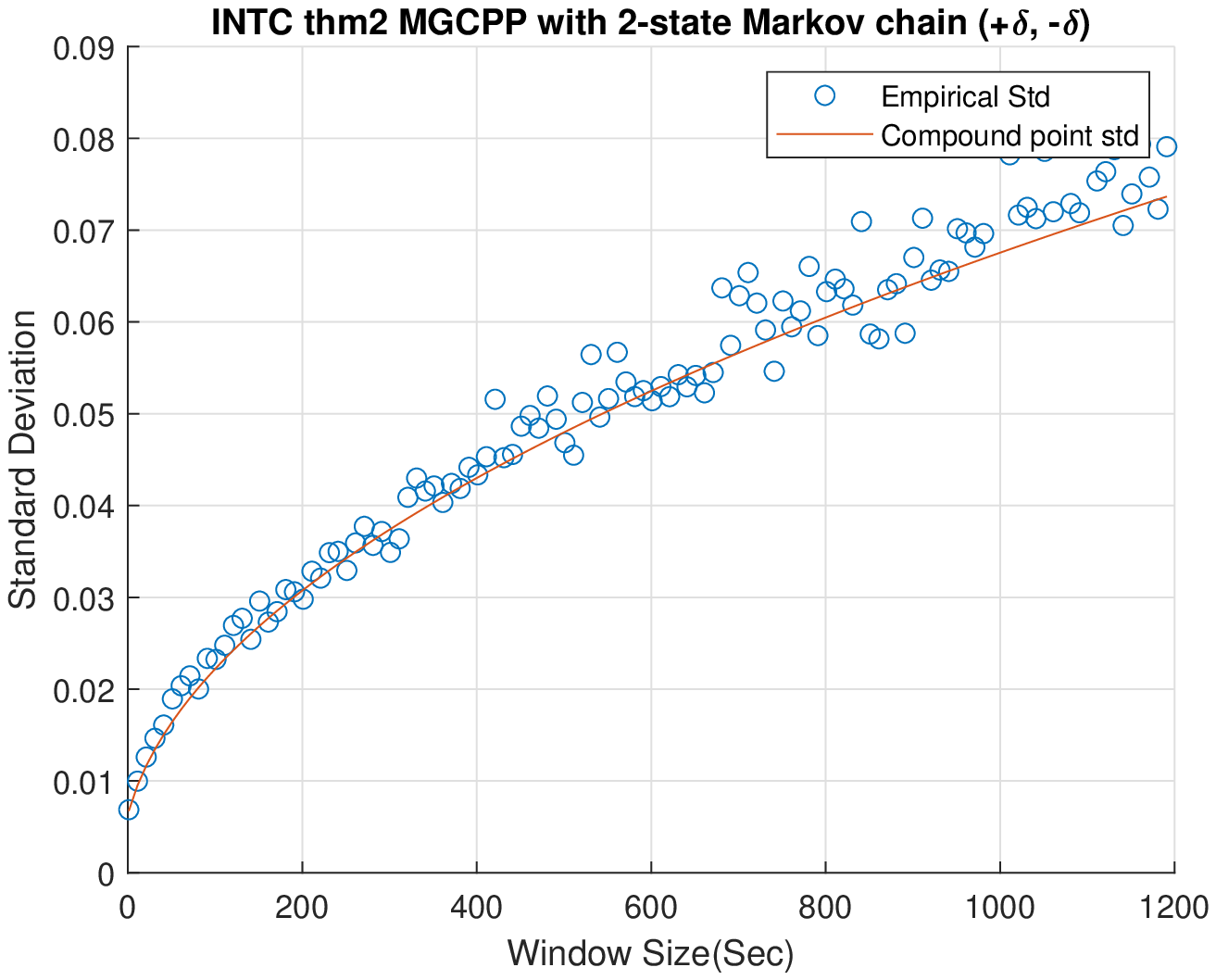}\quad
	\includegraphics[width=.4\textwidth]{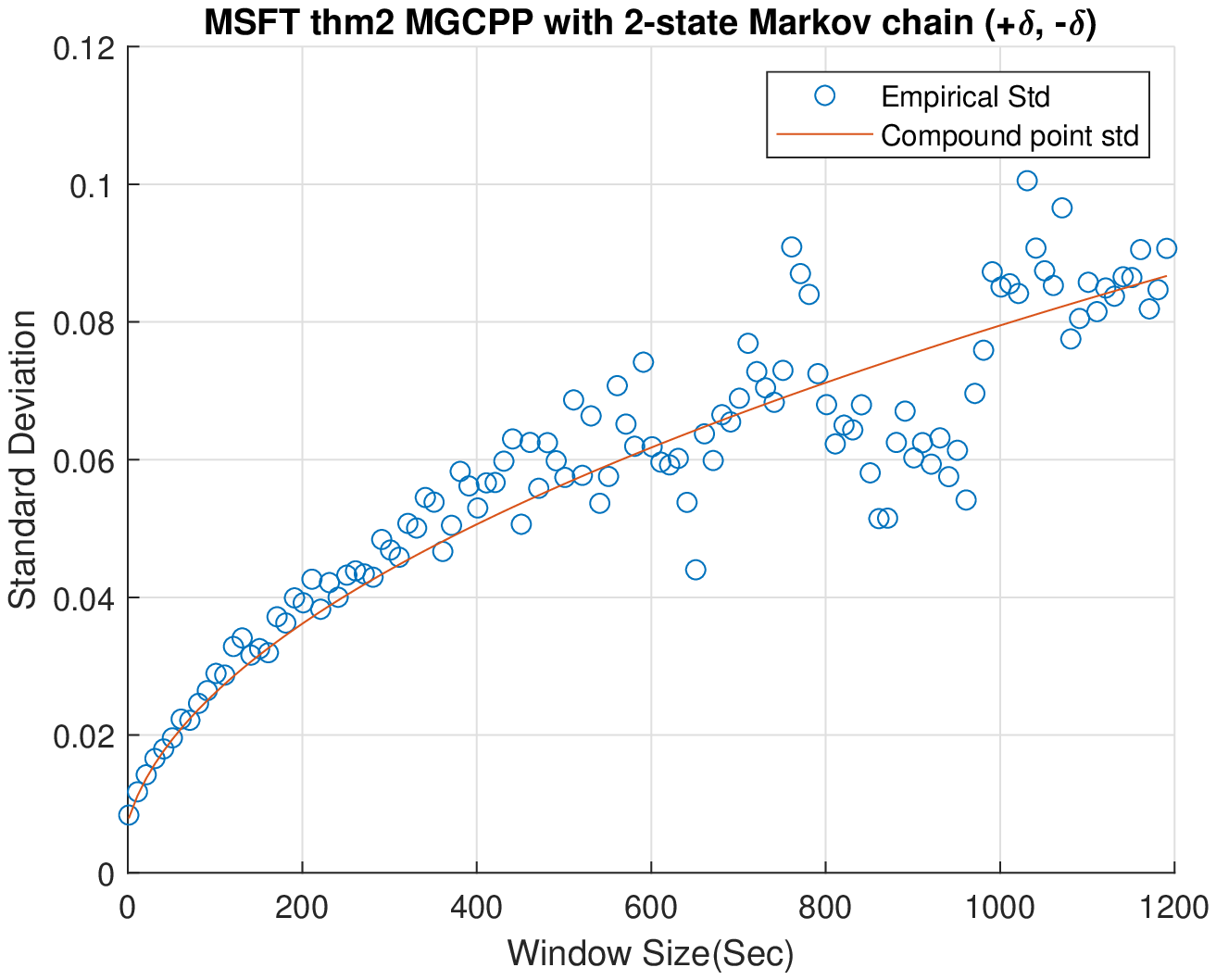}
	\medskip
	\includegraphics[width=.4\textwidth]{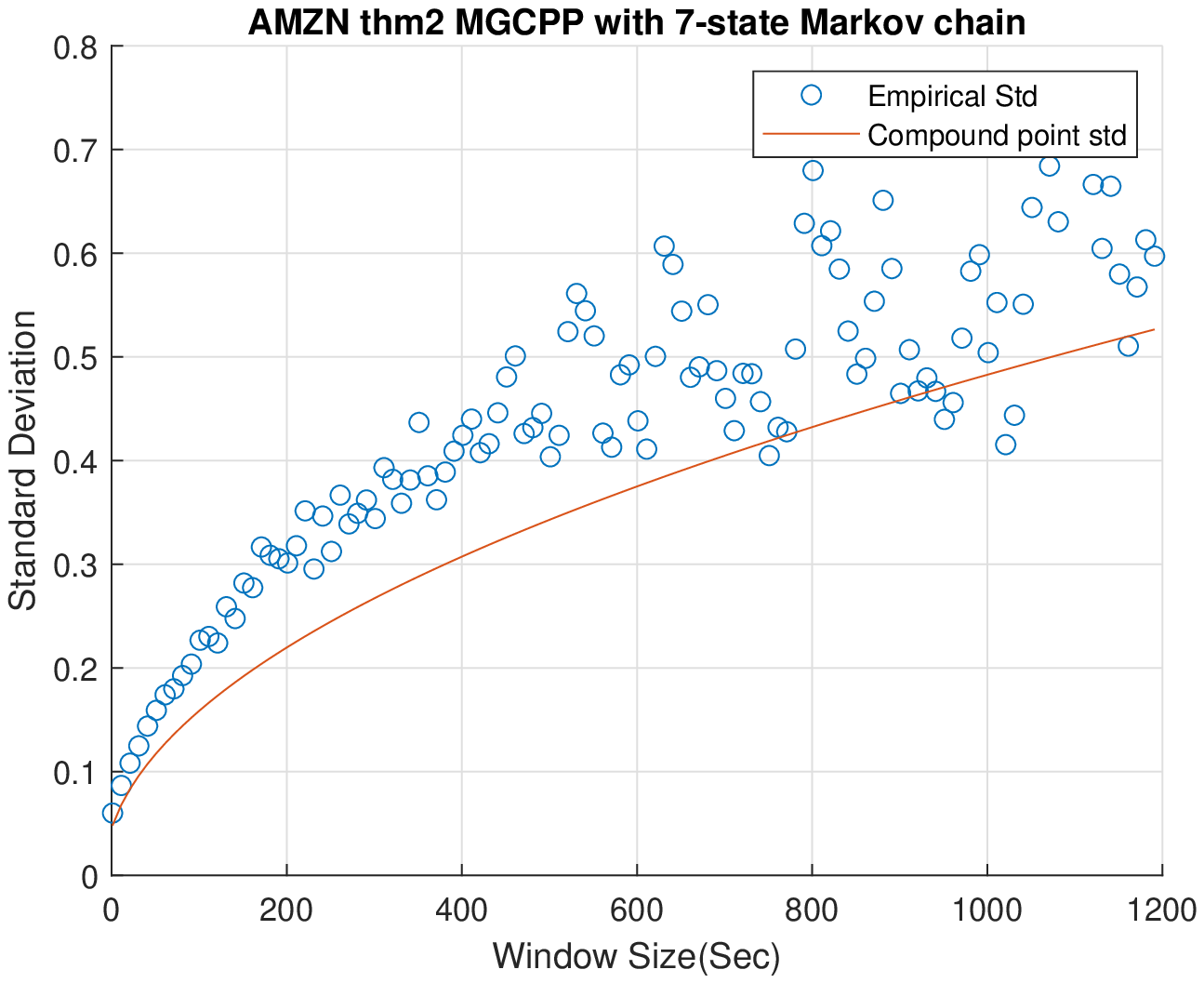}
	\caption{Standard deviation comparisons for 5 stocks by FCLT \RomanNumeralCaps{2} for the MGCPP. INTC and MSFT are simulated with $2$-state Markov chain while AAPL, AMZN, and GOOG are using $7$-state Markov chain.}
	\label{stdfclt2}
\end{figure}

\begin{table}[H]
	\centering
	\caption{The MSE and coefficients computed by MGCPP FCLT \RomanNumeralCaps{2}.}
	\begin{tabular}{ccccc}
		\toprule
		\textbf{Ticker}       & \textbf{MSE}         & \textbf{Regression Ceofficient}  & \textbf{MGCPP Ceofficient}  & \textbf{Percentage Error}  \\ \midrule
		INTC 2-state & 1.5820$\times 10^{-5}$ & 0.0022                 & 0.0021            & $6.27\%$         \\ 
		MSFT 2-state & 6.5788e$\times 10^{-5}$ & 0.00252                & 0.00249           & $0.94\%$         \\ 
		AAPL 7-state & 0.0060     & 0.0278                 & 0.0288            & $3.56\%$         \\ 
		GOOG 7-state & 0.0081     & 0.0307                 & 0.0287            & $6.67\%$         \\ 
		AMZN 7-state & 0.0121     & 0.0189                 & 0.0147            & $22.14\%$        \\  \midrule
		Overall Percentage Error & \multicolumn{4}{c}{$7.92\%$}                             \\ \bottomrule
	\end{tabular}
	\label{coeff2}
\end{table} 

We see that the percentage errors of MSFT and AAPL are very small (less than $5\%$) and the results of INTC and GOOG are also good (less than $10\%$). The percentage error of AMZN is large, but it is still smaller than the error derived from FCLT \RomanNumeralCaps{1} in Table \ref{coeffnstate}. In general, the simulation results of FCLT \RomanNumeralCaps{2} is as good as the FCLT \RomanNumeralCaps{1} and we can apply this FCLT \RomanNumeralCaps{2} to model a mid-price.

\subsection{Rolling Cross-Validation}

In this Section, we tested the forecast ability of the MGCPP model. Since we didn't assume the multivariate point process $\vec{N}_t$ is stationary or independent, we cannot apply the $K$-fold cross-validation directly. Here, we used the rolling $K$-fold cross-validation method which proposed in \cite{Bergmeir2018}. We divided the last $50$ minutes' data into 5 disjoint $10$-min windows for each stock. For the fold $1$, We take the first $280$ minutes' data as the training set to estimate parameters. And then, we applied the data in the next $10$-min window to calculate the percentage error. Next, we merge the test set into the training set in fold $1$ as the new training set in fold $2$ and apply the next $10$-min window as a new test set. Repeat this procedure 5 times, we will get 5 percentage errors. The mean value of the 5 percentage errors will be the test error $E$ for this stock. So, the overall test error for our multivariate model is the average of all test errors. Figure \ref{rolling} gives an example diagram for the rolling cross-validation. 

\begin{figure}[H]
	\includegraphics[width=\linewidth]{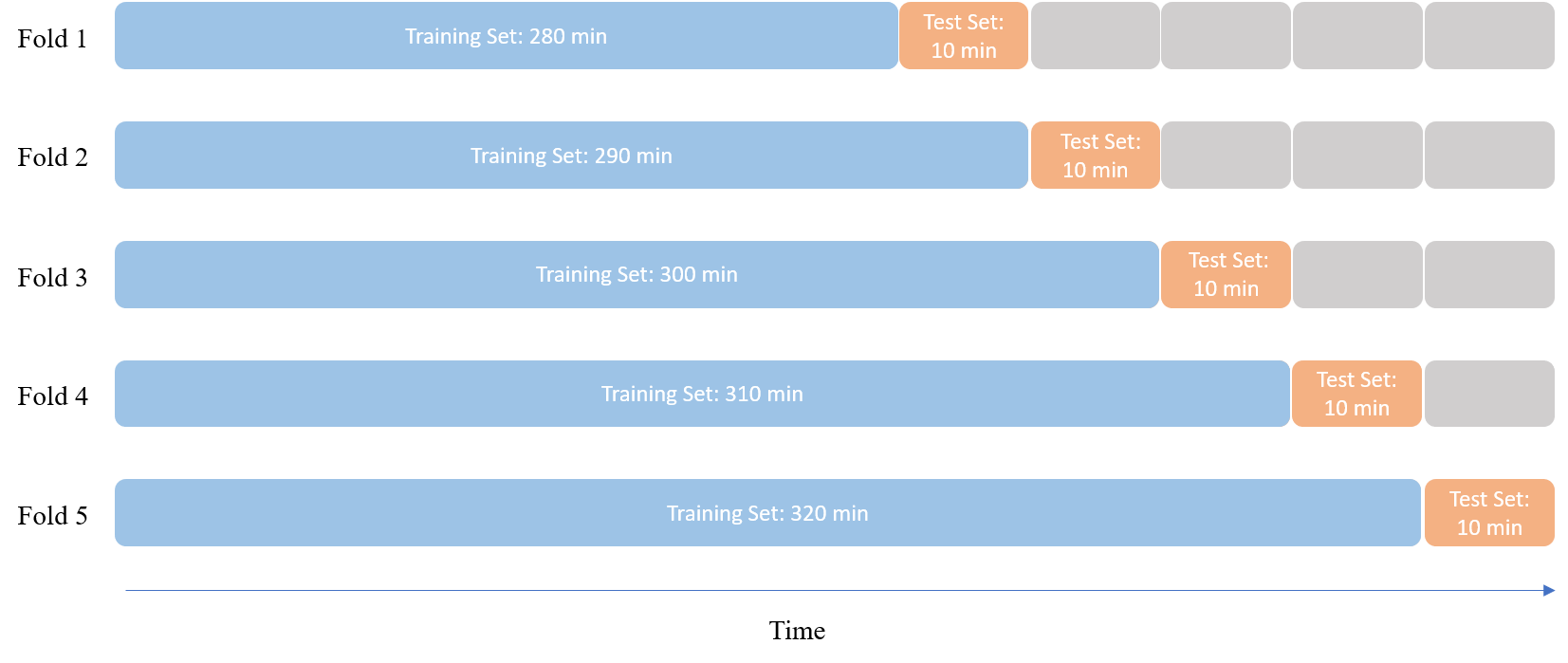}
	\caption{Diagram for the Rolling cross-validation.}
	\label{rolling}
\end{figure}

\begin{table}[H]
	\centering
	\caption{Test Errors for different tickers by applying $5$-fold cross-validation. The errors are percentage errors between regression coefficients and the MGCPP coefficients.}
	\begin{tabular}{ccccccc}
		\toprule
		Ticker             & fold 1    & fold 2    & fold 3    & fold 4    & fold 5    & Mean Error \\ \midrule
		INTC               & $6.75\%$  & $0.39\%$  & $3.16\%$  & $14.32\%$ & $16.60\%$ & $8.24\%$   \\ 
		MSFT               & $20.33\%$ & $31.35\%$ & $16.96\%$ & $8.33\%$  & $22.61\%$ & $19.92\%$  \\ 
		AAPL               & $8.22\%$  & $0.51\%$  & $22.53\%$ & $21.34\%$ & $23.33\%$ & $15.01\%$  \\ 
		GOOG               & $19.60\%$ & $20.41\%$ & $16.41\%$ & $6.13\%$  & $12.51\%$ & $15.19\%$  \\ 
		AMZN               & $20.78\%$ & $4.87\%$  & $7.98\%$  & $18.81\%$ & $42.15\%$ & $18.92\%$  \\ \midrule
		Overall Test Error & \multicolumn{6}{c}{$E_{test}=15.46\%$}                               \\ \bottomrule
	\end{tabular}
	\label{foreerror}
\end{table}

Table \ref{foreerror} lists test errors for different tickers and the overall test error for the MGCPP model. As can be seen from the Table, the test error for each stock is relatively large and the overall test error $(15.46\%)$ is nearly double the overall percentage error $(7.92\%)$ in Table \ref{coeff2}. That's because the results in Table \ref{coeff2} is a fitting error while the test errors in Table \ref{foreerror} is a kind of forecast error. We didn't apply any future information when we conduct the forecast task. So, even the $15.46\%$ overall test error is not as good as the fitting one, it is still a good prediction in the LOB and can provide lots of insights in the forecast task.

\section{Conclusion and future work}\label{sec5}

In this paper, we proposed a multivariate general compound point process for the mid-price modeling in limit order book. This kind of process is a generalization of several stochastic models in the limit order market. We applied LOBSTER data to conduct simulations and found the multivariate generalized model is as good as the general compound Hawkes process model. We also tested the prediction ability of this kind of process. In general, the MGCPP performs very good in LOB modeling and it can be a meaningful reference in the mid-price prediction. In the future, we will explore more applications of the MGCPP and consider related option pricing problems under this kind of frame work.



\begin{thebibliography}{999}
	
	\bibitem{Bacry2013}
	Bacry, E., Delattre, S., Hoffman, M. and Muzy, J.-F. Some limit theorems for Hawkes processes and application to financial statistics. {\em Stochastic Processes and their Applications.} {\bf 2013}, v. 123, No. 7, pp. 2475-2499.
	
	
	\bibitem{Bowsher2007}
	Bowsher, C. Modelling security market events in continuous time: intensity based, multivariate point process models. {\em J. Econometrica.} {\bf 2007}, 141 (2), pp. 876-912.
	
	\bibitem{Bremaud1996}
	Br\'{e}maud, P., and Massoulié, L. Stability of nonlinear Hawkes processes. {\em The Annals of Probability.} {\bf 1996}, 1563-1588.
	
	
	\bibitem{Bjork2011}
	Bjork, T. {\em Introduction to Point Processes from a Martingale Point of View}; KTH, 2011.
	
	\bibitem{Bauwens2009}
	Bauwens, L. and Hautsch, N. {\em Modelling Financial High Frequency Data Using Point Processes}. Springer, 2009.
	
	\bibitem{Bergmeir2018}
	Bergmeir, C., Hyndman, R. J., and Koo, B. A note on the validity of cross-validation for evaluating autoregressive time series prediction. {\em Computational Statistics and Data Analysis.} {\bf 2018}, 120, 70-83.
	
	\bibitem{Cartea2015}
	Cartea, \'{A}., Jaimungal, S. and Penalva, J. {\em Algorithmic and High-Frequency Trading}. Cambridge University Press, 2015.
	
	\bibitem{Chavez2017}
	Chávez-Casillas, J. A., Elliott, R. J., Rémillard, B., and Swishchuk, A. V. A level-1 limit order book with time dependent arrival rates. {\em Methodology and Computing in Applied Probability.} {\bf 2017} 21(3), 699-719.
	
	\bibitem{Chen2019}
	Chen S., Shojaie, A., Shea-Brown E. and Witten D. The Multivariate Hawkes Process in High Dimensions: Beyond Mutual Excitation. {\em arXiv:1707.04928v2.} June 18, 2019.
	
	\bibitem{Cont2013}
	Cont, R. and de Larrard, A. A Markovian modelling of limit order books. {\em SIAM J. Finan. Math.} {\bf 2013} 4(1), pp. 1-25.
	
	\bibitem{Daley2007}
	Daley, D. J., and Vere-Jones, D. {\em An introduction to the theory of point processes: volume II: general theory and structure}. Springer Science and Business Media, 2007.
	
	\bibitem{Eichler2017}
	Eichler, M., Dahlhaus, R., and Dueck, J. Graphical modeling for multivariate hawkes processes with nonparametric link functions. {\em Journal of Time Series Analysis.} {\bf 2017}, 38(2), 225-242.
	
	
	\bibitem{Embrechts2011}
	Embrechts, P., Liniger, T. and Lin, L. Multivariate Hawkes processes: an application to financial data. {\em J. Appl. Prob.} {\bf 2011}, 48, A, pp. 367-378.
	
	\bibitem{Guo2020}
	Guo, Q., and Swishchuk, A. Multivariate general compound Hawkes processes and their applications in limit order books. {\em Wilmott magazine.} {\bf 2020}, 107, 42-51.
	
	\bibitem{Lemonnier2017}
	Lemonnier, R.,  Scaman, K. and Kalogeratos, A. Multivariate Hawkes Processes for Large-Scale Inference. In Proceedings of Thirty-First AAAI Conference on Artificial Intelligence, bf 2017.
	
	\bibitem{Liniger2009}
	Liniger, T. {\em Multivariate Hawkes Processes}. PhD thesis, Swiss Fed. Inst. Tech., Zurich, 2009.
	
	\bibitem{Norris1998}
	Norris, J. R. {\em Markov Chains}. Cambridge University Press, 1998.
	
	\bibitem{Rambaldi2017}
	Rambaldi, M., Bacry, M. and Lillo, F. The role of volume in order book dynamics: a multivariate Hawkes process analysis. {\em Quant. Finance.} {\bf 2017} , v17, 2017, issue 7.
	
	\bibitem{Skorokhod1982}
	Skorokhod, A. {\em Studies in the theory of random processes (Vol. 7021)}, Courier Dover Publications, 1982.
	
	\bibitem{Swishchuk20171}
	Swishchuk, A. Risk model based on compound Hawkes process. Abstract, IME 2017, Vienna.
	
	\bibitem{Swishchuk20172}
	Swishchuk, A. and Vadori, N. A semi-Markovian modelling of limit order markets. {\em SIAM J. Finan. Math.} {\bf 2017}, v.8, pp. 240-273.
	
	\bibitem{Swishchuk20173}
	Swishchuk, A., Cera, K., Hofmeister, T. and Schmidt, J. General semi-Markov model for limit order books. {\em Intern. J. Theoret. Applied Finance.} {\bf 2017} , v. 20, 1750019.
	
	\bibitem{Swishchuk2020}
	Swishchuk, A., and Huffman, A. General compound Hawkes processes in limit order books. {\em Risks.} {\bf 2020}, 8(1), 28.
	
	\bibitem{Vinkovskaya2014}
	Vinkovskaya, E. {\em A point process model for the dynamics of LOB}. PhD thesis, Columbia University, 2014.
	
	\bibitem{Vadori2015}
	Vadori, N., and Swishchuk, A. Strong law of large numbers and central limit theorems for functionals of inhomogeneous Semi-Markov processes. {\em Stochastic Analysis and Applications.} {\bf 2015} , 33(2), 213-243.
	
	\bibitem{Yang2017}
	Yang Y., Etesami J., He N. and Kiyavash N. Online Learning for Multivariate Hawkes Processes. In Proceedings of 31st Conference on Neural Information Processing Systems, Long Beach, CA, USA, 2017.
	
	\bibitem{Zheng2014}
	Zheng, B., Roueff, F. and Abergel, F. Ergodicity and scaling limit of a constrained multivariate Hawkes process. {\em SIAM J. Finan. Math.} {\bf 2014} , 5.
	
	\bibitem{Zhu2013}
	Zhu, L. Central limit theorem for nonlinear Hawkes processes, {\em J. Appl. Prob.} {\bf 2013}, 50(3), pp. 760-771.
	
	
	
\end{thebibliography}

\end{document}